\documentclass[journal]{ieeetran}

\usepackage{pstricks}
\usepackage{mathrsfs}
\usepackage{enumerate}
\usepackage{bbm}
\usepackage{tikz}
\usepackage{pgfplotstable}
\usepackage{algorithm}
\usetikzlibrary{arrows}
\usepackage{relsize}
\usepackage{subcaption}
\usepackage{pgfplots}
\usepackage{scrextend}
\usepackage{epsfig}
\usepackage[cmex10]{amsmath}
\usepackage{amssymb}
\usepackage{amsmath}
\usepackage{bm}
\usepackage{cite}
\usepackage{algpseudocode}
\usepackage{graphicx}
\usepackage{lipsum}
\usepackage{tikz}
\usepackage{mathrsfs} 
\usepackage{tikz}
\usepackage{url}
\usepackage[cmex10]{amsmath}
\usepackage{amssymb}
\usepackage{amsthm}
\usepackage{amsfonts}
\usepackage{bm}
\usepackage{cite}
\usepackage{algpseudocode}
\usepackage{graphicx}
\usepackage{tikz} 
 \usepackage{scrextend} 
 \usepackage[bottom]{footmisc}

\newtheorem{theorem}{Theorem}
\newtheorem{corollary}{Corollary}
\newtheorem{lemma}{Lemma}
\newtheorem{remark}{Remark}
\newtheorem{claim}{Claim}
\newtheorem{definition}{Definition}

\def\pbb {\mathbb{P}}
 \DeclareMathOperator*{\dist}{dist}

\DeclareMathOperator*{\argmax}{argmax}
\DeclareMathOperator*{\ints}{int}
\DeclareMathOperator*{\len}{length}
\begin{document}

\date{}

\title{The Capacity of Associated Subsequence Retrieval
}
\author{Behrooz~Tahmasebi, Mohammad~Ali~Maddah-Ali, and Seyed~Abolfazl~Motahari 
\thanks{
Behrooz Tahmasebi is with the Department of Electrical Engineering and Computer Science (EECS), Massachusetts Institute of Technology (MIT), Cambridge, MA, USA (e-mail: bzt@mit.edu).
Mohammad Ali Maddah-Ali is with the Department of Electrical Engineering, Sharif University of Technology, Tehran 11365, Iran (e-mail: maddah\_ali@sharif.edu). Seyed~Abolfazl~Motahari is with the Department of Computer Engineering, Sharif University of Technology, Tehran 11365, Iran  (e-mail: motahari@sharif.edu).}
\thanks{
 This paper has been presented  at  IEEE ISIT 2018 \cite{early}.
}

 }
 \renewcommand\footnotemark{}
\maketitle

\begin{abstract} 
The objective of a genome-wide association study (GWAS) is to associate  subsequences of individuals' genomes to the  observable characteristics called phenotypes (e.g., high blood pressure). 
Motivated by the GWAS problem, in this paper we introduce the information-theoretic  problem of \emph{associated subsequence retrieval}, where a dataset of $N$ (possibly high-dimensional) sequences of length $G$, and their corresponding    observable (binary)  characteristics is given. The sequences are chosen  independently  and uniformly  at random from $\mathcal{X}^G$,  where $\mathcal{X}$ is a finite  alphabet. 
The  observable  (binary) characteristic  is only related  to  a specific unknown subsequence of length $L$ of the sequences, called   \textit{associated subsequence}. 
For each sequence, if the associated subsequence of it belongs to a universal finite set, then it is more likely to display  the observable  characteristic (i.e., it is more likely that the   observable characteristic is  one).  The goal is to retrieve the associated subsequence using a dataset of $N$ sequences and their  observable  characteristics. We demonstrate that as the parameters $N$, $G$, and $L$ grow, a threshold effect appears in the curve of probability of error versus the rate which is defined as ${Gh(L/G)}/{N}$, where  $h(\cdot)$ is the binary entropy function. This effect allows us  to define the capacity of associated subsequence retrieval.  We develop an achievable scheme and a matching converse for this problem, and thus characterize its capacity  in two scenarios: the zero-error-rate and the $\epsilon$-error-rate. 

%
%
\end{abstract}



  

 
\section{Introduction}
{ In a genome-wide association study (GWAS), the ultimate goal is to find common variants within a population which are associated with a complex disease or a given phenotype. }This task can be fulfilled by sampling individuals from the population and characterizing their variants and disease status at the same time. Due to the high-dimensionality of genomes and complexity of the association, large number of samples are required to retrieve the associated variants reliably. Fortunately, advances in DNA sequencing and microarray technologies have dramatically decreased the cost of information gathering and made such studies a routine procedure in many centers across the globe. 

{
The first step in a GWAS experiment is to sequence or genotype a set of samples from the genomes population. }There exist many tools to reconstruct the genome or identify variants based on raw data, c.f., \cite{alg1,alg2,alg3,alg4,alg5,alg6}. The next step is to infer biological connections between genomic loci and observable characteristics or  phenotypes.  The objective is to find a subsequence of length $L$, called  \emph{associated subsequence},  of a genome of length $G$, which correlates with the observed phenotype. 
In this direction, a  fundamental question is how many individuals are required  to be sampled to retrieve the associated subsequence reliably. 

GWAS has been studied  extensively  and novel  biological results  have been discovered (see e.g.,   \cite{gwas1, gwas4,gwas5,gwas6,gwas9,gwas8,gwas10,gwas11,gwas12,gwas13, gwas14}). 
An important application of GWAS is when the observed phenotype is related to a disease.
There are a number of works studying the associated subsequences of  diabetes (type I  \cite{gwas4}  and type II \cite{gwas5}) and  various types of cancers \cite{gwas6},  e.g.,  the breast cancer \cite{gwas9} and the prostate cancer \cite{gwas8}.

Recently, a number of  computational biology problems have been studied from an information-theoretic viewpoint. 
For example, in the problem of  reconstructing a genome from the sequencing reads, information-theoretic limits are  characterized in \cite{motahari1}, followed by \cite{guy, motahari2, motahari3, tse1, ilan1,ilan2,ilan3}. Also, a number of authors considered the DNA storage systems, and studied the capacity and coding designs for them in several scenarios  \cite{newit2, newit3, newit4, newit5, newit6, ilan4, ilan5}. Haplotype assembly is another example which falls into this category \cite{hap, newit1}. 

 Motivated by the GWAS problem,  in this paper, we introduce an abstract information-theoretic problem of \emph{associated subsequence retrieval}; see Fig. \ref{fig:model}.  
In this problem, a dataset of $N$ sequences of length $G$
 and their corresponding  observable (binary) characteristics is given\footnote{In this paper, we assume that the dataset is drawn from one population and it is homogeneous.}.
  The given sequences  are denoted by $\bold{x}_n$, $n=1,2,\ldots,N$, and  are chosen independently and uniformly at random from $\mathcal{X}^G$, for a finite alphabet $\mathcal{X}$. 
There is a subsequence $\bold{s}=(s_1,s_2, \ldots ,s_L)$ of  $(1,2,\ldots,G)$, unknown a priori, which is associated with a  binary characteristic. In the model, if $\bold{x}_\bold{s}:=(x_{s_1}, x_{s_2}, \ldots ,x_{s_L})$ equals to one of the $m\in \mathbb{N}$ possible a priori unknown sequences, then the observable characteristic is more likely to be one. The set of $m$ unknown sequences is denoted by $\mathcal{Q}$. Therefore,  if $\bold{x}_\bold{s}\notin \mathcal{Q}$ then the corresponding observable characteristic is more likely to be zero.  In the considered information-theoretic problem, we would like to study the asymptotic behavior of the minimum sample complexity  $N$ such that $\bold{s}$  can be reliably identified from the given $N$ sequences and their corresponding observable characteristics. 



\begin{figure*}\begin{center}
\includegraphics[scale = 1]{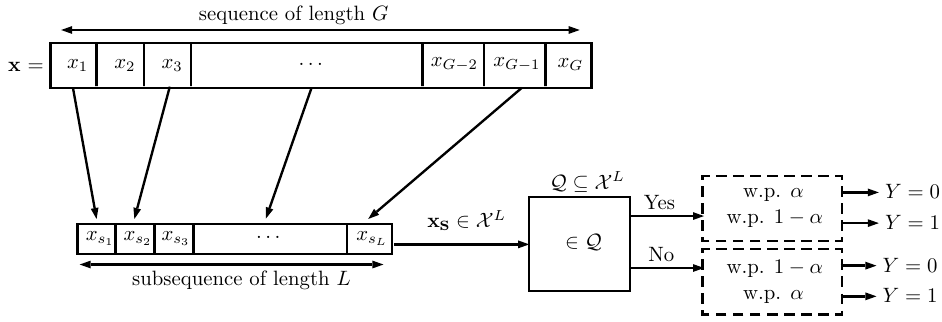}
\caption{The problem of associated subsequence retrieval. For each sequence, a specific unknown subsequence of it is chosen, and if the result belongs to a specific universal set $\mathcal{Q}$, then it is more likely to display the characteristic. The main objective  is to retrieve the unknown subsequence $\bold{s} = (s_1,s_2,\ldots, s_L)$  from the given $N$ i.i.d. samples,  without relying on any prior information about $\mathcal{Q}$.  } \label{fig:model}

\end{center}
\end{figure*}

The main contribution of this paper is to show that there is  a threshold effect in the error probability of the associated subsequence retrieval problem. In particular, we define the rate of problem  
as $\frac{G h(L/G)}{N}$, where $h(\cdot)$ is the binary entropy function. We prove that in the asymptotic regimes, if the rate is less than a  threshold, then there exist retrieval schemes with arbitrary low probabilities of error. Conversely, if the rate is above that  threshold, then there is no scheme having vanishing  probability of error. The threshold is then called the capacity of  associated subsequence retrieval, and it is  explicitly characterized it in this paper.


In particular, we define two notions of  the zero-error-rate and the $\epsilon$-error-rate   associated subsequence retrieval. The  error-rate  is defined  as  the fraction of incorrectly associated indices in the estimation of the associate subsequence. 
 For the $\epsilon$-error-rate estimation, we are interested in  retrieving the  associated subsequence with an error-rate of at most $\epsilon$, while for the zero-error-rate estimation, no positive error-rate is acceptable. 
In the zero-error-rate case, we fully characterize the capacity. The capacity is proven to be a finite positive number, which shows that the  scaling of parameters in the definition  of    rate is reasonable. 
  In the $\epsilon$-error-rate regime, we show that for small enough $\epsilon$, 
 the capacity is 
  the same  as the
    zero-error-rate case. 
 This shows that the two problems of the zero-error-rate and  the $\epsilon$-error-rate associated subsequence retrieval are equivalent in the asymptotic regimes. 
 
 The rest of the paper is organized  as follows.  Section \ref{Problem Formulation} is devoted to the mathematical model and the  definition of the capacity. In Section \ref{Main Results},   the main results of the paper are presented. The proofs are available in Section  
 \ref{Proof1} and Section \ref{Proof2}, and finally, Section \ref{Conclusion} concludes the paper.

\section{ Problem Statement}\label{Problem Formulation}

\subsection{Notation}
In this paper,   random variables are denoted by capital letters, such as $X$, and their realizations are denoted by lower case letters, such as $x$. 
But as an exception, we use the capital letters $N$, $G$ and $L$ to denote the problem's parameters, which are non-random.
For a (discrete) random variable $X$, $p_X$ denotes its  probability mass function. 
Random probability mass functions are also denoted by capital letters, like $P_X$. 
The  sequences are denoted  by bold letters, like $\bold{x}$,  and the random sequences are denoted by capital bold letters, like $\bold{X}$. 
For any positive integer $G$, let $[G]:=\{1,2,\ldots,G\}$. 
The $\ell_p$ norm of a vector $\bold{w} \in \mathbb{R}^n$ for  $p \ge 1$ is defined as
\begin{align}
\| \bold{w} \|_{p} := \Big(  \sum_{i=1}^n |w_i|^p \Big )^{1/p}.
\end{align}
For a sequence $\bold{x}=(x_1,x_2,\ldots, x_G )$  of  length $G$,  and a sequence $  \bold{s} = (s_1, s_2, \ldots ,s_L) \in [G]^L$  of length $L$,   we define $\bold{x}_{\bold{s}}:= (x_{s_1},  x_{s_2}, \ldots  x_{s_L})$. $\bold{x}_{\bold{s}}$  is clearly a  subsequence of $\bold{x}$. 
Also, we denote the length of a sequence $\bold{x}$ by $\len(\bold{x})$. 
The base two logarithm is denoted by $\log(.)$.
For any $p \in [0,1]$, the binary entropy function $h:[0,1] \to [0,1]$ is defined as
\begin{align}
h(p) :=  p  \log(\frac{1}{p})+(1-p)  \log(\frac{1}{1-p}).
\end{align}
The mutual information of two discrete random variables $X$ and $Y$ is denoted by $I(X;Y).$ 
 The distance of two sequences $\bold{s}=(s_1 ,s_2 ,\ldots ,s_L)$ and  $\bold{t}=(t_1, t_2 ,\ldots ,t_L)$ is defined as 
 \begin{align}
\dist(\bold{s}, \bold{t}):= \Big | \{s_1,s_2, \ldots, s_L\}\triangle\{t_1,t_2, \ldots, t_L\}\Big |,
\end{align}
where  $\triangle$ denotes the symmetric difference of sets.
   The set of all strictly increasing subsequences of length $L$  of  $(1,2,\ldots,G)$ is denoted by $\mathcal{S}_{L,G}$. 
   Given a finite set $\mathcal{X}$ and an integer $m$, we define  
   \begin{align}
   \mathcal{F}_{L,m} := \Big \{ f :  \mathcal{X}^L \rightarrow \{0,1\} : ~     \Big | f^{-1}(1)\Big |  = m    ~  \Big \}, \label{fclass}
\end{align}
where $f^{-1}(1)$ is the set of sequences  $\bold{x}\in \mathcal{X}^L$, such that $f(\bold{x}) = 1$. In other words, $\mathcal{F}_{L,m}$ is the set of functions that exactly map  $m$ sequences of $\mathcal{X}^L$ to one, and map the others  to zero. 
   
\subsection{System Model}

Consider a dataset of $N$ i.i.d. samples $\bold{x}_1,\bold{x}_2,\ldots, \bold{x}_N$, where each sample $\bold{x}_n$ is chosen uniformly at random from the set of  sequences of length $G$ from a finite alphabet  $\mathcal{X}$, i.e., $\bold{x}_n \in \mathcal{X}^G$. The observed characteristic of each sample $\bold{x}_n$, denoted by $y_n$, can only take two states\footnote{
For simplicity, in this paper, we only consider binary characteristics. However, the results and proofs of this paper are valid for any  characteristic with a finite label set. 
} denoted by 0 and 1.


There is a stochastic map $\mathscr{F}: y_n = \mathscr{F}(\bold{x}_n)$,  which associates the $n^{\text{th}}$ sample  $\bold{x}_n \in \mathcal{X}^G$ to its observed characteristic $y_n\in \{0,1\}$. 
As depicted in Fig. \ref{fig:model}, the map $\mathscr{F}$ is formed as follows: first a specific unknown subsequence  of length $L$ of $\bold{x}$,  denoted by $\bold{x}_{\bold{s}}=(x_{s_1},  x_{s_2}, \ldots, x_{s_L})$, is chosen. 
Then, $\bold{x}_{\bold{s}}$ goes through an indication  function that outputs 1 if $\bold{x}_{\bold{s}}\in \mathcal{Q}$, and zero otherwise, for a universal set $\mathcal{Q}$, with size $|\mathcal{Q}|=m$, $m\in \mathbb{N}$. As shown in Fig. 2, we denote this indication function by $f(.)$. Note that the parameter $m$ denotes the number of patterns in the associated subsequence which increase the probability of displaying the characteristic (i.e., the probability of $Y_n=1$). Then, a Bernoulli random variable is  XOR'd with    $f(\bold{x}_{\bold{s}})$    
and the result is the observed characteristic $y_n$. 
More precisely, $y_n=f(\bold{x}_{n,\bold{s}}) \oplus Z_n$, where $Z_n$ is a Bernoulli random variable with parameter $\alpha$. 
We assume that the additive noises $\{Z_n\}_{n\in \mathbb{N}}$ are    independent from the sequences and are  independently chosen for different samples. 
The existence of  the additive noise in the problem setup represents the effect of   the other factors in the observed characteristic, such as the environmental effects, which are not related to the genome sequences in the GWAS problem.  In one extreme, the observed characteristics are  highly correlated with the sequences ($\alpha \approx 0$). In the other extreme,  the labels are approximately independent from the sequences ($\alpha \approx 1/2$). We assume in this paper that $\alpha \in [0,1/2)$ is   given.

In the model, the sequence $\bold{s}$ and the deterministic function $f(.)$ are unknown but they are  the same for all  $N$ sequences. Throughout  this paper, we call $\bold{s}$ the associated subsequence, and the main objective is to retrieve  $\bold{s}$. 
It is  assumed that the parameter $L$ (the length of the associated subsequence) is given.  In addition, in this paper, we focus on cases that $L \ll G$. More precisely, we assume that  $L/G$ goes to zero, whenever we consider the asymptotic regimes in the paper. This assumption is motivated by what we observed in the GWAS problem.    
We also assume that the deterministic function $f(.)$ is chosen uniformly at random from $\mathcal{F}_{L,m}$, defined in (\ref{fclass}),  for a given positive integer $m$.   
  The associates subsequence $\bold{s}$ is selected randomly and uniformly  from the set of all strictly increasing sequences of length $L$ with the entries belonging to $[G]$, which is  denoted by $\mathcal{S}_{L,G}$
  \footnote{ 
  We note that because there is no information about $\bold{s}$ and $f(.)$ in the model, we assume  that the prior distribution of them is uniform.  
  It is also worth mentioning that the uniform sampling of $\bold{s}$ is only needed for the converse proofs. The achievability  proofs hold for any prior distribution on the set $\mathcal{S}_{L,G}$. Also, the uniform sampling of $f(.)$ is only required for the achievability proof and the converse proof holds for any prior on $\mathcal{F}_{L,m}$.
  }. 
The entries of the sequence $\bold{s}$ represent the sites in each sampled sequence that affect the observed characteristic.

\begin{figure*}
\begin{center}
\includegraphics{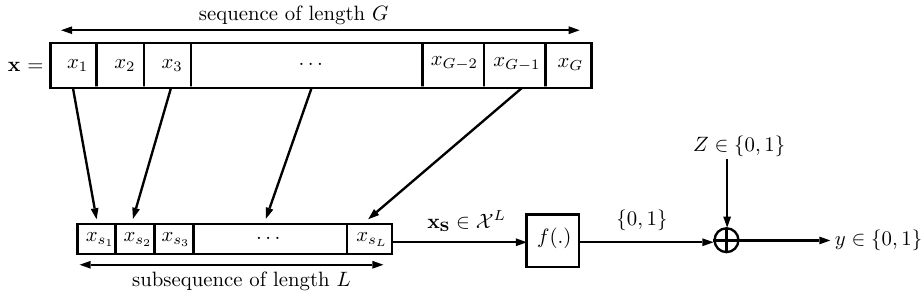}
\caption{The system model. For any  sequence $\bold{x}\in \mathcal{X}^G$, a subsequence  of it with length $L$  is selected. 
The function $ f(.)$ is one for $m$ sequences and zero for the other $|\mathcal{X}^L| -m$ sequences.   Finally, $y= f(\bold{x}_{\bold{s}}) \oplus Z$, where $Z\sim $ {Bern}$(\alpha)$, and $\alpha \in [0,0.5)$. 
} 
\end{center}
 \label{fig:model2}
\end{figure*}

Let us define the parameter $\beta:= \pbb (Y_n=1)$.  Note that  $ \beta \in (\alpha,1-\alpha)$ in the proposed data generation model. It is assumed that $\beta$ is known\footnote{
Note that  $\beta$ is a constant that only depends on $m,L,\alpha$ and $|\mathcal{X}|$.
}. 
In this paper, the objective is to   estimate $\bold{s}$, given $N$ sampled sequences $\{\bold{x}_n\}_{n\in [N]}$ and their corresponding observed characteristics $\{y_n\}_{n \in [N]}$. 
In the following, we formally define the algorithms for this purpose.

\begin{definition}
Algorithm $\mathcal{A}_{(G,L,N,{\alpha},\beta,m)}$  is a mapping from the set of all possible input datasets, $(\mathcal{X}^G)^N\times \{0,1\}^N$,  to the set  $\mathcal{S}_{L,G}$. When  there is no ambiguity, we denote an algorithm by $\mathcal{A}_{G}$ or  $\mathcal{A}$. For a dataset $(\{\bold{x}_n\}_{n \in [N]},\{y_n\}_{n \in [N]})$, $\hat{\bold{s}}=\mathcal{A}(\{\bold{x}_n\}_{n \in [N]},\{y_n\}_{n \in [N]})$ denotes  the output of the algorithm. 
\end{definition}

Next we formally define the  error event and also the probability of error for an algorithm.

\begin{definition}
For a positive  $\epsilon $ and  an algorithm $\mathcal{A}_{(G,L,N,{\alpha},\beta,m)}$, the error event $\mathcal{E}_{\mathcal{A}, \epsilon}$ is defined as  
 $\mathcal{E}_{A,\epsilon}:=\{\frac{\dist(\bold{\hat{S}} ,\bold{S} )}{L}>\epsilon\}$, where  $\hat{\bold{S} }$ is the  output of the algorithm.  Also, the worst-case probability of error of an  algorithm $\mathcal{A}$ is defined as\footnote{In this paper, the required condition to retrieve  the associated subsequence is  defined based on bounding the probability of having a large  error-rate, with respect to a threshold. However, one can see that if the desired condition is defined based on vanishing/bounding the  expectation of the error rate, that is $\mathbb{E} \{\frac{\dist(\bold{\hat{S}} ,\bold{S} )}{L} \}$, the same results hold on the capacity of the problem.}  
\begin{align}
P^{\text{WC}}_{\epsilon}(\mathcal{A}):=  \max_{\bold{s} \in  \mathcal{S}_{L,G}} \pbb(\mathcal{E}_{\mathcal{A}, \epsilon}| \bold{S}=\bold{s})\label{pbbnote}.
\end{align}
The average probability of error is also defined as $P^{\text{AVG}}_{\epsilon}(\mathcal{A}):= \pbb(\mathcal{E}_{\mathcal{A},\epsilon})$.
\end{definition}

 \begin{remark} \normalfont
 Note that the probability measure $\pbb$ in (\ref{pbbnote}) is defined  with respect to the  random  sequences $\bold{X}_1,\bold{X_2},\ldots,\bold{X}_N$, the random noises $Z_1,Z_2,\ldots,Z_N$,  the random subsequence $\bold{S} \in \mathcal{S}_{L,G}$, and the random function $F(.) \in \mathcal{F}_{L,m}$. 
  \end{remark}

 \begin{remark}   \normalfont
The parameter $\epsilon$ is a threshold for the normalized distance between  $\bold{s}$ and its estimation    $\hat{\bold{s}}$. 
 Note that for any algorithm $\mathcal{A}$,  $P^{\text{AVG}}_{\epsilon}(\mathcal{A}) \le P^{\text{WC}}_{\epsilon}(\mathcal{A})$.
Also, by the definition of the error event,   if $\epsilon_1 \ge \epsilon_2$ then $\mathcal{E}_{\mathcal{A},\epsilon_1} \subseteq \mathcal{E}_{\mathcal{A}, \epsilon_2}$ and thus $P^{\text{WC}}_{\epsilon_1}(\mathcal{A})\le P^{\text{WC}}_{\epsilon_2}(\mathcal{A})$ and $P^{\text{AVG}}_{\epsilon_1}(\mathcal{A})\le P^{\text{AVG}}_{\epsilon_2}(\mathcal{A}).$
 \end{remark}   

In this paper,  the goal is to characterize the fundamental limits of the associated subsequence retrieval problem, i.e.,  the region for the parameters of the  problem such that the retrieval of $\bold{s}$ is possible. 
For this purpose, we derive the fundamental limits in two scenarios. First, we study the problem in the  zero-error-rate regime, meaning that no positive error-rate is allowed. Second, we study the problem of approximating   $\bold{s}$ such that a  positive  error-rate of at most $\epsilon$ is acceptable. 
In the following definitions, first we  define the achievable  algorithms and rates and then we  define  the capacity  of problem. 

\begin{definition}
For any positive $\epsilon$,  a sequence of algorithms $\{\mathcal{A}_{(G_i,L_i,N_i,{\alpha}, \beta, m_i)}\}_{i \in \mathbb{N}}$\footnote{
Assume that $G_i,N_i$ and $L_i$ are strictly increasing functions of $i$.
}, where\footnote{
By  $m=o(N)$, we mean ${m_i}/{N_i} \to 0$ as $i \to \infty$. More precisely, this means  that the size of the given dataset $N$  is much greater than the number of   patterns which increase the probability of  displaying the characteristic.
} $m=o(N)$, is  said to be  $\epsilon-$achievable, if and only if     $P^{\text{WC}}_{\epsilon}(\mathcal{A}_{G_i}) \rightarrow 0$ as $i \rightarrow \infty.$ 
\end{definition}

\begin{definition}
A positive real $R$ is said to be an $\epsilon-$achievable rate, if and only if  there is an $\epsilon-$achievable    sequence of algorithms $\{\mathcal{A}_{(G_i,L_i,N_i,{\alpha}, \beta, m_i)}\}_{i \in \mathbb{N}}$, such that 
$R \le \frac{G_i h(L_i/G_i)}{N_i}$ for each $i$, where
 $h(.)$ is the binary entropy function. 
\end{definition}

\begin{definition}
A positive real number $R$ is  said to be achievable, if and only if for any positive $\epsilon$,   $R$ is $\epsilon-$achievable.
\end{definition}

Now we are ready to define the capacity  of the problem.

\begin{definition}
The zero-error-rate capacity  is defined as the supremum of all achievable rates  and is denoted by $C(\alpha,\beta).$ 
Also, for any positive $\epsilon$,  the  $\epsilon-$capacity is defined as the supremum of all $\epsilon-$achievable rates and is denoted by $C_{\epsilon}(\alpha,\beta)$.
\end{definition}


\begin{remark}   \normalfont
Due to the definitions of the capacities,  
\begin{align}
 C({\alpha},\beta) = \inf_{\epsilon>0} C_{\epsilon}({\alpha},\beta) \le C_{\epsilon_1}({\alpha},\beta) \le C_{\epsilon_2}({\alpha},\beta),
\end{align}
for any  positive real numbers  $\epsilon_1,\epsilon_2$,  such that  $\epsilon_1\le \epsilon_2$. 
\end{remark}

 \section{Main Results}\label{Main Results}
 
In this section, we state the main results of the paper. In the following theorem, we characterize the  capacity of the associated subsequence retrieval  $C(\alpha,\beta).$

\begin{theorem}\label{thrm1}
The zero-error-rate capacity of the associated subsequence retrieval is 
\begin{align}
C({\alpha},\beta)={h(\beta)-h({\alpha})}, 
\end{align}
where $h(.)$ is the binary entropy function. 
\end{theorem}

The achievability  proof of the theorem is  in Section \ref{achievability}, and the proof of the converse can be found in Section \ref{converse}. 
\begin{remark} \normalfont
  Theorem \ref{thrm1}  characterizes the  capacity of the  zero-error-rate associated subsequence  retrieval. This shows that there is a threshold effect in the problem at  $Gh(L/G)/N$, in the asymptotic regimes.  
  Note that  the capacity is strictly positive since $ \beta \in (\alpha,1-\alpha)$.
  \end{remark}
  
  \begin{remark} \normalfont
  For the achievability, we examine all  the  subsequences of length $L$ of the given sequences  and choose the one for which the two binary vectors $(f(\bold{x}_{n ,\hat{\bold{s}}}))_{n \in [N]}$ and $(y_n)_{n \in [N]}$ are jointly typical for some $f \in \mathcal{F}_{L,m}$.
  Note that unlike the channel coding, there is no codebook in this setup and   the sequences are  produced by nature. This changes  the proof techniques. Still,  we prove that the probability of error in the proposed scheme goes to zero, using the  approximation methods ignoring the dependency among some events and bounding the effect of this assumption.
  \end{remark}
  
  \begin{remark} \normalfont
  For the converse, we cannot directly use   Fano's inequality due to the definition of the error event.  Instead,  we develop some inequalities similar to it. 
  The need for this new bound is due to the fact that in this case, if there is an  approximation of the associated subsequence with the error-rate of at most $\epsilon$, then we cannot determine it exactly. This fact yields  some new terms which appear in the converse proof that are  required to be investigated.
  \end{remark}

We are also interested in characterizing the minimum number of required samples  to find an {approximation}  of the associated subsequence with respect to a  given positive error-rate $\epsilon$.  In the following theorem, we state the result of this paper on the $\epsilon$-error-rate capacity.

\begin{theorem}\label{thrm2}
 There is a positive $\epsilon_0 \in (0,1/2)$, such that for any $\epsilon \in (0,\epsilon_0)$, 
\begin{align}
C_{\epsilon}({\alpha},\beta)={h(\beta)-h({\alpha})}.
\end{align}
\end{theorem}

The proof of Theorem \ref{thrm2} can be found in Section \ref{Proof2}.

\begin{remark}  \normalfont
It  may be surprising that the $\epsilon-$capacity is the same as the zero-error-rate capacity. This shows that there is no difference between the approximation of the associated subsequence  and the zero-error-rate retrieval in the asymptotic regimes, from the perspective of   sample complexity. 
\end{remark}

\begin{remark} \normalfont
To prove Theorem \ref{thrm2}, we develop a complementary procedure to convert any algorithm that approximates the associated subsequence, according to an error rate of at most $\epsilon$,  to another algorithm that retrieves it with the zero-error-rate condition. 
\end{remark}

\section{Proof of Theorem \ref{thrm1}}\label{Proof1}
\subsection{Achiveability }\label{achievability}
Let $R  < h(\beta) - h(\alpha)$ be a positive real number. We aim  to prove that $R$ is achievable. In particular, for any positive $\epsilon$, we want to show that $R$ is $\epsilon-$achievable. To this end, let us  introduce an algorithm  achieving this rate. 
The algorithm is a jointly typical decoder. First,    a few definitions are required.

Let $F(\bold{X}_{\bold{S}})$ denote the output of the model  depicted in Fig. 2,  where  $\bold{X},\bold{S}$, and $F(.)$ are uniform random instances from $\mathcal{X}^G$, $\mathcal{S}_{L,G}$, and $\mathcal{F}_{L,m}$,  respectively. 
Note that 
\begin{align}
p_{(F(\bold{X}_{\bold{S}}),Y)} (a,b)= &\Big(\gamma \mathbbm{1}\{a=1\}+ (1-\gamma)  \mathbbm{1}\{a=0\}   \Big ) \times \\
 &  \Big( \alpha \mathbbm{1} \{ a \neq b\} + (1-\alpha) \mathbbm{1} \{ a = b\}\Big),
\end{align} 
for any $a,b \in \{0,1\}$, where $\gamma:= m/|\mathcal{X}|^L$.
Also, due to the model,   
\begin{align}
p_{(F(\bold{X}_{\bold{S}}),Y)} (.,.)= p_{(F(\bold{X}_{\bold{s}}),Y)}(.,.) = p_{(f(\bold{X}_{\bold{s}}),Y)}(.,.),
\end{align} 
for any $f(.) \in \mathcal{F}_{L,m}$ and any $\bold{s} \in \mathcal{S}_{L,G}$.

Now for any positive $\tau$, let  $\mathcal{T}^{N}_{\tau}$  denote the set of all  jointly  typical binary sequences of length $N$, with respect  to  $p_{(F(\bold{X}_{\bold{S}}),Y)}$. More precisely, 
\begin{align*}
&\mathcal{T}_{\tau}^{N}:= \Big \{    (\bold{u},\bold{v}) \in \{0,1\}^N \times \{ 0 , 1\}^N : \\
& \Big |-\frac{1}{N} \log(\prod_{n=1}^N p_{F(\bold{X}_{\bold{S}})}(u_n)) - H(F(\bold{X}_{\bold{S}})) \Big |  < \tau  ,\\
 &\Big |-\frac{1}{N} \log(\prod_{n=1}^N p_{Y}(v_n)) - H(Y) \Big | < \tau , \\
& \Big |-\frac{1}{N} \log(\prod_{n=1}^N p_{(F(\bold{X}_{\bold{S}}),Y)}(u_n,v_n)) -H(F(\bold{X}_{\bold{S}}),Y) \Big | < \tau  \Big \}.
\end{align*}

In other words, $\mathcal{T}^N_{\tau}$ includes the pairs of binary sequences of length $N$ with empirical entropies $\tau-$close to the true entropies with respect to    $p_{(F(\bold{X}_{\bold{S}}),Y)}$
\cite{cover}. 
Note that  we only require  the parameters $\alpha,\gamma $  to test whether $(\bold{u},\bold{v}) \in \mathcal{T}_{\tau}^N$ or not\footnote{We just use the function $F(.)$  (or a non-random $f(.)$) and the sequence $\bold{S}$ (or a non-random $\bold{s}$) in the indices in notations to remind the definition of them in the proofs. }.


The proposed algorithm is as follows.

\textit{Algorithm:}  For the given  dataset $( \{\bold{x}_n\}_{n\in [N]}, \{y_n\}_{n\in [N]})$, the algorithm chooses  $\hat{\bold{s}} \in \mathcal{S}_{L,G}$  with the following property: there is a function $f(.)  \in \mathcal{F}_{L,m}$, 
 such that 
the binary vectors $(f(\bold{x}_{n ,\hat{\bold{s}}}))_{n \in [N]}$ and $(y_n)_{n \in [N]}$ are jointly typical, i.e.,  
$\Big ((f(\bold{x}_{n ,\hat{\bold{s}}}))_{n \in [N]},(y_n)_{n \in [N]}\Big ) \in \mathcal{T}^N_{\tau}$.
  If there are more than one such $\hat{\bold{s}}$ with this property, the algorithm chooses one of them at random. If there is no such $\hat{\bold{s}}$, the algorithm chooses an element of $\mathcal{S}_{L,G}$  at random. We denote the proposed algorithm by $\mathcal{A}_G.$

\begin{remark} \normalfont
Note that the typicality decoder also achieves an estimation of the function $f(.)$. But since  in this paper we only require the asymptotically vanishing error-rate  (not the probability  $\mathbb{P}(\bold{\hat{S}} \neq \bold{S})$), we cannot bound the probability of error for the estimation of $f(.)$ in the typicality decoder. We also emphasize that the estimation of $f$ can be independently  carried out after the estimation of $\bold{s}$, and this is known to be a  simple  learning problem since the VC dimension of the class of functions considered in the paper is bounded ($m \ll N$).  
\end{remark}

 \textit{Analysis of the algorithm:}
 For any fixed positive $\epsilon$, we aim to prove that any $R < h(\beta) - h(\alpha)$ is $\epsilon-$achievable using the proposed algorithm. In particular, we are interested to show that $P^{WC}_{\epsilon}(\mathcal{A}_G) \rightarrow 0$ as $N,G,L \rightarrow \infty$. In other words, for any $\bold{s} \in \mathcal{S}_{L,G}$, we want to prove that the probability
  $\pbb(\frac{\dist(\hat{\bold{S}},\bold{S})}{L} > \epsilon  | \bold{S} =\bold{s}) $ goes to zero in the asymptotic regimes.

 Fix an arbitrary $\bold{s} \in \mathcal{S}_{L,G}$. 
 Consider the two events $\mathcal{E}_1$ and $\mathcal{E}_2$ as follows. $\mathcal{E}_1$ is the event that the associated subsequence $\bold{s}$ does not satisfy the   acceptance properties  of  the algorithm, i.e.,
 \begin{align}
\mathcal{E}_1 := \Big \{  \Big ((F(\bold{X}_{n ,{\bold{S}}}))_{n \in [N]},(Y_n)_{n \in [N]}\Big ) \not \in \mathcal{T}^N_{\tau} \Big |
\bold{S} = \bold{s}
 \Big \}.
\end{align}

 Also $\mathcal{E}_2$  is the event that there exists at least one $\bold{t} \in \mathcal{S}_{L,G}$ such that $\bold{t}$ satisfies the acceptance properties of  the algorithm and $\dist(\bold{s},\bold{t})> L \epsilon$.
 To be more precise, let us  define
 \begin{align}
\mathcal{E}_{\bold{t},g} := \Big \{
\Big ((g(\bold{X}_{n ,{\bold{t}}}))_{n \in [N]},(Y_n)_{n \in [N]}\Big ) \in \mathcal{T}^N_{\tau}
\Big |
\bold{S} = \bold{s}
\Big \}.
\end{align}
and let
 \begin{align}
\mathcal{E}_{2} :=  
\bigcup_{
\substack{
            \bold{t} \in \mathcal{S}_{L,G}\\
            g(.) \in \mathcal{F}_{L,m}\\
            \dist(\bold{s},\bold{t})> L \epsilon
		}
} \mathcal{E}_{
\bold{t},g
}.
\end{align}
 
 Note that to  prove  the $\epsilon-$achievability of $R$,  it just suffices to show that $\pbb(\mathcal{E}_1  \cup \mathcal{E}_2 ) \rightarrow 0$. 
 Using the union of events bound, $\pbb(\mathcal{E}_1  \cup \mathcal{E}_2 ) \le \pbb(\mathcal{E}_1)+\pbb( \mathcal{E}_2 )$ and thus it suffices to show that   $\pbb(\mathcal{E}_1)$  and $\pbb(\mathcal{E}_2)$ vanish in the asymptotic regimes. 
 We note that    $\pbb(\mathcal{E}_1)$ vanishes in the asymptotic regimes, using the law of large numbers. Therefore,  to complete  the proof, it suffices to show that $\pbb(\mathcal{E}_2)$ vanishes  asymptotically.
 
 Let us state a lemma. 
  
  \begin{lemma}\label{lemma1}\cite[Theorem 3.5]{fondml} (Sauer's lemma)
 For any positive integers $N,m$ and any ${\bold{t}}\in \mathcal{S}_{L,G}$,  
  \begin{align*}
 \max_{\forall n\in [N]: \bold{x}_n \in   \mathcal{X}^G}&\Big | \Big \{(g(\bold{x}_{n ,{\bold{t}}}))_{n \in [N]} \in \{0,1\}^N  :  g(.) \in \mathcal{F}_{L,m} \Big \} \Big | \\
 & \le   \sum\limits_{i=0}^d {   N \choose  i},  
\end{align*}
where $d$ is the VC dimension\footnote{Vapnik-Chervonenkis dimension.} of the class of functions $\mathcal{F}_{L,m}$.
  \end{lemma}
 
   \begin{corollary}\label{corollary2}
   For any positive integers $N \ge m$ and any $\bold{t} \in \mathcal{S}_{L,G}$, we have
   \begin{align*}
   \max_{\forall n\in [N]: \bold{x}_n \in   \mathcal{X}^G}&\Big | \Big \{(g(\bold{x}_{n ,{\bold{t}}}))_{n \in [N]} \in \{0,1\}^N :  g(.) \in \mathcal{F}_{L,m} \Big \} \Big | \\
   & \le (\frac{eN}{m})^m.
   \end{align*}
   \end{corollary}
   \begin{IEEEproof}
  The VC dimension of the class of functions $\mathcal{F}_{L,m}$ can be upper bounded\footnote{
   Indeed, it can be shown that $d = \min(m,|\mathcal{X}|^L-m).$ 
   } by $m$. Therefore, using Lemma \ref{lemma1} and \cite[Corollary 3.3]{fondml}, we can establish the desired result. 
   \end{IEEEproof}
   
Corollary \ref{corollary2} relates the parameter $m$ (which is related to the VC dimension of the class of functions considered) to  the number of possible observable patterns, given $N$ instances of a function.
This bound is used in the 
asymptotic analysis of the algorithm $\mathcal{A}_G$ by the union bound.
Note that the condition of Corollary \ref{corollary2} is satisfied, since   $m=o(N)$.

  \begin{theorem}\label{lemma2}
 For any  positive  real numbers $\zeta, \epsilon$, and any $\bold{s} \in \mathcal{S}_{L,G}$,  with  probability $1-o(1)$,  the following proposition holds:
 \begin{enumerate}[]
\item  
For any   $\bold{t} \in \mathcal{S}_{L,G}$ and any function $g(.) \in \mathcal{F}_{L,m}$, such that  $\dist(\bold{s},\bold{t})> L\epsilon$,  the probability   that  $\bold{t}$  satisfies  the acceptance conditions in the proposed algorithm  via the function $g(.)$ 
(i.e., the probability of the event $ \mathcal{E}_{
\bold{t},g
}$)
is upper bounded by $2^{-N  (h(\beta) - h(\alpha)  - \zeta)  }$. 
\end{enumerate}
  \end{theorem}

\begin{IEEEproof}
See Appendix \ref{appendA}. 
\end{IEEEproof}

Using Corollary \ref{corollary2}, to  analyze the  algorithm $\mathcal{A}_G$, it just suffices to check at most $(eN/m)^m$ functions.
Let us denote the event in the statement of Theorem \ref{lemma2}, which holds with probability $1-o(1)$,  by $\mathcal{E}_3$.
   Now  we write
\begin{align}
\pbb(\mathcal{E}_2) & = \pbb(\mathcal{E}_2|\mathcal{E}^c_3) \pbb(\mathcal{E}^c_3) + \pbb(\mathcal{E}_2|\mathcal{E}_3) \pbb(\mathcal{E}_3)  \\& \le \pbb(\mathcal{E}^c_3) + \pbb(\mathcal{E}_2 | \mathcal{E}_3)    \\
&   \le o(1) + \pbb \Big(\bigcup_{
\substack{
            \bold{t} \in \mathcal{S}_{L,G}\\
            g(.) \in \mathcal{F}_{L,m}\\
            \dist(\bold{s},\bold{t})> L \epsilon
		} 
} \mathcal{E}_{\bold{t},g} ~\Big |~ \mathcal{E}_3  ~ \Big)     \\
& \overset{(a)}{\le} o(1) +  \sum_{
\substack{
            \bold{t} \in \mathcal{S}_{L,G}\\
            \dist(\bold{s},\bold{t})> L \epsilon
		}
}(\frac{eN}{m})^m \times 2^{-N (h(\beta) - h(\alpha)  - \zeta )}  \\
&  \le o(1) +  {G \choose L}  \times 2^{  m\log (
\frac{eN}{m})  -N(h(\beta) - h(\alpha)  - \zeta ) }  ,
\end{align}
where (a) follows by Corollary \ref{corollary2}, Theorem \ref{lemma2},    and also   the union bound. 
Using \cite[Chapter  11, p.~353]{cover},  we have $ {G \choose L} \le 2^{ G h(L/G)}$. Therefore,  
\begin{align}
\pbb(\mathcal{E}_2) &  \le o(1)  + 2^{  G  h(L/G) +  m \log (\frac{eN}{m})  -N(h(\beta) - h(\alpha)  - \zeta )}     \\
  &=  o(1) + 2^{  N  \big  (R +  \frac{m}{N} \log (\frac{eN}{m})   -(h(\beta) - h(\alpha) -  \zeta ) \big)} .
\end{align}
We note that 
$\frac{m}{N} \log (\frac{eN}{m}) \to 0,$ since $m = o(N)$.
Thus,  $\pbb(\mathcal{E}_2)$ vanishes asymptotically if 
\begin{align}
R - h(\beta) - h(\alpha) -\zeta <0.
\end{align}  
This shows that by choosing  small enough $\zeta$,  any $R  < h(\beta) - h(\alpha)$ is $\epsilon-$achievable. This holds for any positive $\epsilon$ and thus completes the proof.

\subsection{ Converse Proof }\label{converse}
In this section, we prove the converse part of  Theorem \ref{thrm1}. First we state a lemma.

\begin{lemma}\label{lemm1}
 For any positive $\epsilon \in (0,1/2)$, let $R$ be an $\epsilon-$achievable rate. Then,  
 \begin{align}
	R \le \frac{h(\beta)-h(\alpha)}{1-h(\epsilon)}.
\end{align} 
  \end{lemma}

\begin{IEEEproof}
See appendix \ref{prooflemm3}.
\end{IEEEproof}
  
Now consider  an achievable rate $R$. By definition, for any positive $\epsilon$, $R$ is $\epsilon-$achievable. Using Lemma \ref{lemm1}, we conclude that  $R\le \frac{h(\beta)-h(\alpha)}{1-h(\epsilon)},$ for any positive $\epsilon.$ Therefore,   
\begin{align}
R \le \inf_{\epsilon \in (0,1/2)} \frac{h(\beta)-h(\alpha)}{1-h(\epsilon)}=h(\beta)-h(\alpha),
\end{align} 
and this completes the proof.

\section{Proof of Theorem \ref{thrm2}}\label{Proof2}

The achievability proof of Theorem \ref{thrm2} directly follows from Theorem \ref{thrm1}.
 Therefore, only the converse proof is required.

Let that $R$ be an $\epsilon-$achievable rate. The goal is to show that $R \le h(\beta) - h(\alpha).$ First, according to Lemma \ref{lemm1}, $R \le \frac{h(\beta) - h(\alpha)}{1-h(\epsilon)}$. Consider a sequence of algorithms $\{\tilde{\mathcal{A}}_{G_i}\}_{i \in \mathbb{N}}$, such that $P^{WC}_{\epsilon}(\tilde{\mathcal{A}}_{G_i}) \rightarrow 0$. Define a complementary procedure as follows. 
Denote the output of the algorithm $\tilde{\mathcal{A}}_{G_i}$ by $\tilde{\bold{s}}$. 
Let $\mathcal{B}_{\tilde{\bold{s}}, \epsilon}$ be the ball with radius  $L\epsilon$ in $\mathcal{S}_{L,G}$ around $\tilde{\bold{s}}$, with respect to the distance  $\dist(.,.)$.
More precisely, we define 
\begin{align}
\mathcal{B}_{\tilde{\bold{s}}, \epsilon} :=  \Big \{ 
\bold{t} \in \mathcal{S}_{L,G} : \dist(\tilde{\bold{s}},\bold{t}) \le L\epsilon 
\Big \}.
\end{align}

 Now we apply the proposed algorithm in the achievability proof of Theorem \ref{thrm1} to find   $\hat{\bold{s}} \in \mathcal{B}_{\tilde{\bold{s}}, \epsilon}$ as the estimation of $\bold{s}$.
More precisely, find the one $\hat{\bold{s}} \in \mathcal{B}_{\tilde{\bold{s}}, \epsilon}$ such that 
\begin{align}
\Big ((f(\bold{x}_{n ,\hat{\bold{s}}}))_{n \in [N]},(y_n)_{n \in [N]}\Big ) \in \mathcal{T}^N_{\tau},
\end{align}
for some $f(.) \in \mathcal{F}_{L,m}$.
If there is not such $\hat{\bold{s}}$, choose one of the elements of $\mathcal{B}_{\tilde{\bold{s}}, \epsilon}$ randomly. If there is more than one such element, choose one of them randomly.

We aim to prove that the output of this complementary algorithm is a zero-error-rate estimation of $\bold{s}$, with high probability. In particular, for any $\epsilon' \in (0,\epsilon)$, we show that the output of the complementary procedure has at most $\epsilon'$-error-rate, with high probability in the asymptotic regimes. 
More precisely, we prove that
$\pbb(\frac{\dist(\hat{\bold{S}},\bold{S})}{L} > \epsilon' | \bold{S} =\bold{s}) \to 0 $ for any  $\epsilon' \in (0,\epsilon)$.

Fix a positive $\epsilon' \in (0,\epsilon)$.
Define two events $\mathcal{E}_1$ and $\mathcal{E}_2$ similar to the achievability proof of Theorem \ref{thrm1}.
More precisely,  define
 \begin{align*}
\mathcal{E}_1 := &\underbrace{ \Big \{  \Big ((F(\bold{X}_{n ,{\bold{S}}}))_{n \in [N]},(Y_n)_{n \in [N]}\Big ) \not \in \mathcal{T}^N_{\tau} \Big |
\bold{S} = \bold{s}
 \Big \} }_{\mathcal{E}_{11}} \bigcup  \\
    & \underbrace{ \Big \{  \bold{S} \not \in \mathcal{B}_{\tilde{\bold{S}}, \epsilon} \Big |
\bold{S} = \bold{s}
 \Big \} }_{\mathcal{E}_{12}}  .
\end{align*}
In addition, we define
 \begin{align*}
\mathcal{E}'_{\bold{t},g} :=& \underbrace{ \Big \{
\Big ((g(\bold{X}_{n ,{\bold{t}}}))_{n \in [N]},(Y_n)_{n \in [N]}\Big ) \in \mathcal{T}^N_{\tau}
\Big |
\bold{S} = \bold{s}
\Big \} }_{\mathcal{E}_{\bold{t},g} } \bigcap  \\
& \Big \{
\bold{t} \in \mathcal{B}_{\tilde{\bold{S}}, \epsilon}
\Big |
\bold{S} = \bold{s}
\Big \}  ,
\end{align*}
and
 \begin{align}
\mathcal{E}'_{2} :=  
\bigcup_{
\substack{
            \bold{t} \in \mathcal{S}_{L,G}\\
            g(.) \in \mathcal{F}_{L,m}\\
            \dist(\bold{s},\bold{t})> L \epsilon'
				}
} \mathcal{E}'_{
\bold{t},g
} .
\end{align}
 Note that $\mathcal{E}'_2$ is  defined with respect to the parameter $\epsilon'$. To complete the proof, we need to show that $\pbb(\mathcal{E}_{1} \cup \mathcal{E}'_{2}) \to 0$. 
 Define
\begin{align}
\mathcal{E}_{2} :=  
\bigcup_{
\substack{
\bold{t} \in \mathcal{B}_{\bold{s}, 2\epsilon}\\
            g(.) \in \mathcal{F}_{L,m}\\
            \dist(\bold{s},\bold{t})> L \epsilon'
				}
} \mathcal{E}_{
\bold{t},g
} .\end{align}

\begin{claim}$\mathcal{E}_{1} \cup \mathcal{E'}_{2} \subseteq \mathcal{E}_{1} \cup \mathcal{E}_{2}.$
\end{claim}

\begin{IEEEproof}
The claim directly follows from the definitions.
\end{IEEEproof}
 
 As a result, the proof is complete if 
 $\pbb(\mathcal{E}_{1} \cup \mathcal{E}_{2}) \to 0$ or $\pbb(\mathcal{E}_{1}) , \pbb(\mathcal{E}_{2})   \to 0$.
By the assumption, $\dist(\tilde{\bold{S}},\bold{s}) \le L\epsilon$, with high probability. 
In other words, $\pbb(\mathcal{E}_{12}) \to 0$. By the law of large numbers,    $\pbb(\mathcal{E}_{11}) \to 0$. Using  the union bound,  $\pbb (\mathcal{E}_1) \to 0$.

To complete the proof, we only need to show that  $
\pbb(\mathcal{E}_2)$ goes to zero asymptotically. 
 Similar to the analysis of the proposed algorithm in Theorem \ref{thrm1},  using Corollary \ref{corollary2} and Theorem \ref{lemma2},  we  write\footnote{Defnine $\mathcal{E}_3$ with respect to the parameter $\epsilon'$.}
  \begin{align}
&\pbb(\mathcal{E}_2)  \\&= \pbb(\mathcal{E}_2|\mathcal{E}^c_3) \pbb(\mathcal{E}^c_3) + \pbb(\mathcal{E}_2|\mathcal{E}_3) \pbb(\mathcal{E}_3)  \\& \le \pbb(\mathcal{E}^c_3) + \pbb(\mathcal{E}_2 | \mathcal{E}_3)   \\
&   \le o(1) + \pbb \Big (\bigcup_{
\substack{
           \bold{t} \in \mathcal{B}_{\bold{s}, 2\epsilon}\\
            g(.) \in \mathcal{F}_{L,m}\\
            \dist(\bold{s},\bold{t})> L \epsilon'
		} 
} \mathcal{E}_{\bold{t},g} ~\Big |~ \mathcal{E}_3  ~\Big)    \\
& \le o(1) +    \max_{\bold{s} \in \mathcal{S}_{L,G}}  |\mathcal{B}_{\bold{s}, 2\epsilon}  | \times (\frac{eN}{m})^m \times  2^{-N  (h(\beta) - h(\alpha)  - \zeta )}.\label{eq.1}
\end{align}
Now observe that
\begin{align}
 \log &(|\mathcal{B}_{\bold{s},2\epsilon}| )\\&\le
 \log   \Big  (\sum_{\ell=0}^{\lfloor 2L\epsilon \rfloor}  \Big |  \Big \{  \bold{t}  \in \mathcal{S}_{L,G}: \dist(\bold{s}, \bold{t})=\ell  \Big |  \Big  \} \Big )\\
& \le  \log   \Big (\sum_{\ell=0}^{\lfloor 2L\epsilon \rfloor}{ L \choose \ell}{G-L \choose \ell}  \Big )  \\
& \le  \log  \Big ( (2L \epsilon+1){ L \choose \lfloor 2L\epsilon \rfloor}{G-L\choose \lfloor  2L  \epsilon \rfloor}  \Big ) \\
& \overset{(a)}{\le} \log( (2L \epsilon+1)  \times 2 ^{ L h(2\epsilon)}\times  2^{(G-L)h(\frac{2L \epsilon }{G-L})})\\
&= \log ( 2L \epsilon+1)+ L h(2\epsilon)+ (G-L)h(\frac{2L \epsilon }{G-L})\\
& \overset{(b)}{\le}  \log ( 2L \epsilon+1) + G h(\frac{4\epsilon L}{G})\label{eq.2},
\end{align}
where (a) follows from \cite[Chapter  11, p.~353]{cover} and (b) follows from the concavity of the binary entropy  function $h(.)$. 
Therefore, using (\ref{eq.1}) and (\ref{eq.2}), we write
\begin{align*}
\pbb&(\mathcal{E}_2) \\ & \le o(1) +  \max_{\bold{s} \in \mathcal{S}_{L,G}} |\mathcal{B}_{\bold{s}, 2\epsilon}|  2^{m  \log (\frac{eN}{m})   -N  (h(\beta) - h(\alpha)  - \zeta )}\\
& \le  o(1) + (2L \epsilon+1)   2^{   G h(\frac{4\epsilon L}{G})   + m \log (\frac{eN}{m})    -N  (h(\beta) - h(\alpha)  - \zeta )} \\
&= o(1) +  (2L \epsilon+1)   2^{ N  \big( R  \frac{ h(4\epsilon L/G)}{h(L/G)}   + \frac{m}{N}  \log (\frac{eN}{m}) -(h(\beta) - h(\alpha)  - \zeta ) \big )} \\
& \overset{(a)}{\le} o(1) +   (2L \epsilon+1)  2^{   N \big ( R   h(2\epsilon)   + \frac{m}{N}  \log (\frac{eN}{m}) -(h(\beta) - h(\alpha)  - \zeta ) \big)},
\end{align*}
where (a) follows from Lemma \ref{lemma4}.

Now since $\frac{m}{N} \log (\frac{eN}{m}) \to 0$, 
we conclude that $\pbb(\mathcal{E}_2)$ vanishes asymptotically
\footnote{
The multiplicative factor $2L\epsilon+1$ does not make any problem, noting that $\frac{\log(L)}{N} \rightarrow 0$ asymptotically. 
This is due to the fact that by the definition of the problem $N \gg m  = \Theta (|\mathcal{X}^L|)$. 
}, if 
\begin{align}
	R \le \frac{h(\beta)- h(\alpha) - \zeta}{h(2\epsilon)},
\end{align}
for some positive $\zeta$. 
Let us assume  $h(2\epsilon) +h(\epsilon)< 1$. Using Lemma \ref{lemm1}, we have 
\begin{align}
R \le \frac{h(\beta)- h(\alpha) }{1-h(\epsilon)}  < \frac{h(\beta)- h(\alpha)  - \zeta}{h(2\epsilon)},
\end{align}	
for small enough $\zeta$. Thus,  if $R$ is $\epsilon-$achievable, then it is $\epsilon'-$achievable for any $\epsilon' \in (0,\epsilon)$. This means that $R$  is achievable. Therefore, using Theorem \ref{thrm1},  we conclude that $R \le h(\beta) - h(\alpha)$ and this completes the proof. 
Note that there is a positive $\epsilon_0$ such that for any $\epsilon < \epsilon_0$, we have  $h(2\epsilon) +h(\epsilon)< 1$.

\begin{remark} \normalfont
Numerical calculation shows that $\epsilon_0 \approx 0.075$ works for the converse of Theorem \ref{thrm2}.
 However, we do not claim that this is the optimum threshold. 
\end{remark}

\section{Conclusion and Discussion}\label{Conclusion}
In this paper, the  capacity of the associated subsequence retrieval problem, which is inspired by a biological  data analysis problem known as genome-wide association study (GWAS),  is studied. 
The fundamental limits of the sample complexity of the problem  are derived, for the zero-error-rate and the   $\epsilon$-error-rate regimes. 
In particular, it is shown that the two problems of the $\epsilon$-error-rate and the zero-error-rate associated subsequence retrieval are   equivalent.

For the future work, a number of  problems can be investigated which are listed below.
\begin{itemize}
\item In this paper, it is assumed that the dataset is homogeneous, i.e., the $N$ sampled sequences are associated to only one subsequence. However, it is more realistic to consider the mixed population datasets, where  the sampled sequences belong to more than one sub-population, where each sub-population has a specific associated subsequence and the population origin of the individuals are unknown. This problem has been studied recently in \cite{itw}. 
\item Another direction is to consider  the non i.i.d. sequences. Although the i.i.d. assumption plays an important role in the proofs of this paper, non i.i.d. sequences, such as stationary Markov models can also be explored for this problem. 
\item In this paper, it is assumed that the sampling of $N$ sequences is independent from their observed characteristics. However, it is more realistic  to consider the non-independent sampling. For example in the GWAS problem, there exist some cases  where  a specific phenotype  is  rare, and hence in the independently  sampled datasets, only a  a few proportion of the population display the characteristic. 
\item Another direction for the future work is to consider other notions of the probability of error for the associated subsequence retrieval. For example, it is worth to consider the exact retrieval condition, that is $\pbb( \hat{\bold{S}}\neq \bold{S}) \to 0$.  It can be shown  that if the function $f(.)$ is linear, and the sequences are binary, the capacity of associated subsequence retrieval is the same as the one appeared in this paper  \cite{iwcit}. However, for the general case, the capacity of the problem is unknown. 
\item In this paper, it is proved that there is a threshold effect in the cure of the probability of error for associated subsequence retrieval. A remaining problem is to explore how fast this probability of error goes to zero for the rates below the capacity. This problem, which is  known as the error exponents problem, is another direction for the further studies. 

\end{itemize}

\appendices

\section{Proof of Theorem \ref{lemma2}}\label{appendA}

To prove Theorem \ref{lemma2}, we need   a few  preliminary definitions and lemmas which are available in the following two subsections.

\subsection{Preliminaries}
In this subsection,  we first review some definitions about the divergence measures on probability distributions, as well as their main properties.
\begin{definition}\label{def1}
Let $f:\mathbb{R}^{\ge0} \to \mathbb{R}$ be a convex function, such that $f(1) = 0$ and $f(t)$ is strictly convex at $t = 1$. Then, the $f-$divergence of any   (discrete) probability measures $p_U$ and $q_U$  on a finite set $\mathcal{U}$ is defined as    
\begin{align}
D_f(p_U||q_U) = \sum_{u \in \mathcal{U}} q_U(u) f \Big (\frac{p_U(u)}{q_U(u)} \Big ).
\end{align}
\end{definition}
We notice that the $f-$divergences satisfy the data processing inequality.
\begin{theorem}\label{thrm3} \cite{fdivergence} 
(Data processing inequality for   $f-$divergences). For any   (finite) probability  measures $p_U,q_U$ and any channel $p_{V|U}$, the following inequality holds.
\begin{align}
D_f(p_U||q_U) \ge D_f(p_Up_{V|U}||q_Up_{V|U}).
\end{align}
\end{theorem}

Note that the function $f(t) = \frac{1}{2}|1-t|$ satisfies the required conditions in Definition \ref{def1}. It can be shown that in this case, the $f-$divergence reduces to the total variation distance of two probability measures.

In the following definition, we define the $f-$information of two arbitrary (discrete) random variables. 

\begin{definition}
For any   (discrete) random variables $U$ and $V$, we define
\begin{align}
I_f(U;V) := D_f(p_{U,V}||p_Up_V).
\end{align}
\end{definition}

Specifically, for the case of $f(t) =  \frac{1}{2}|1-t|$, we can write
\begin{align}
I_f(U;V)  = \frac{1}{2}\| p_{U,V} - p_Up_V\|_1.
\end{align}

\begin{lemma}\label{lemmadata}
Consider three random variables $U,V,W$, such that    $U-V-W$ is a Markov chain.
Then, for any $f-$divergence we have
\begin{align}
I_f(U;V) \ge I_f(U;W).
\end{align}
\end{lemma}

\begin{IEEEproof}
First we define the following  channel 
\begin{align}
q_{T,W|U,V} := p_{W|V} p_{T|U}, 
\end{align}
where $p_{T|U}$ is the identity channel, i.e., $T=U$ with probability one. Note that
\begin{align}
I_f(U;V) &=  D_f(p_{U,V}||p_Up_V) \\
&\overset{(a)}{\ge} D_f(p_{U,V}q_{T,W|U,V}||p_Up_Vq_{T,W|U,V})\\
&=  D_f(p_{U,V}p_{W|V} p_{T|U}||p_Up_Vp_{W|V}  p_{T|U})\\
&\overset{(b)}{=}  D_f(p_{U,V}p_{W|U,V} p_{T|U}||p_Up_Vp_{W|V}  p_{T|U})\\
& =   D_f(p_{T,W}||p_Tp_W)\\
& = I_f(U;W),
\end{align}
which completes the proof. Note that   (a) follows from Theorem \ref{thrm3} and (b) follows from the fact that $p_{W|U,V}=p_{W|V}$.
\end{IEEEproof}

\begin{corollary}\label{corollary1}
Consider   random variables $U,V,W,T$, such that    $U-V-W-T$ is a Markov chain.
Then, 
\begin{align}
\| p_{U,T} - p_U p_T\|_{1} \le \| p_{V,W} - p_V p_W\|_{1} .
\end{align}
\end{corollary}

 \begin{IEEEproof}
 Consider $f(t) = \frac{1}{2}|1-t|$ and use Lemma \ref{lemmadata} twice. 
 \end{IEEEproof}

Note that  $\ell_1$ and $\ell_{\infty}$  norms are equivalent.
\begin{lemma}\label{lemmanorm}
For any $\bold{w} \in \mathbb{R}^n$,  
\begin{align}
\| \bold{w} \|_{\infty} \le \| \bold{w} \|_{1} \le n \| \bold{w} \|_{\infty}  
\end{align}
\end{lemma}

In what follows, we state a few  definitions about the dependency of  (discrete)  random variables.
\begin{definition}\label{appind}
For any (discrete) random variables $U,V$ and any   ${\mu} \in [1,\infty)$, we write $U ~\overset{{\mu}}{\bot}~ V$ if and only if  $p_{UV}(u,v)\le {\mu} \times  p_U(u)p_V(v)$ for all $u,v$.
\end{definition}
Note that for any independent random variables $U,V$, we have  $U ~\overset{1}{\bot}~ V$.
Also, if  $ U ~\overset{{\mu}}{\bot}~ V$, then    $U ~\overset{{\mu}'}{\bot}~ V$ for any ${\mu}' \ge {\mu}$.

The following lemma relates the above definition to the total variation distance.
\begin{lemma}\label{lemma6}
Assume that  $\|p_{U,V}-p_Up_V\|_1 \le \epsilon$ for  a  positive $\epsilon$. Then, for 
\begin{align}
{\mu} = 1+\frac{\epsilon}{\min\limits_u p_U(u)\times \min\limits_v p_V(v)},
\end{align}
we have $U ~\overset{{\mu}}{\bot}~ V$.
\end{lemma}
\begin{IEEEproof}
First we note that using Lemma \ref{lemmanorm}, we have \begin{align}
\|p_{U,V}-p_Up_V\|_{\infty}  \le \|p_{U,V}-p_Up_V\|_1\le \epsilon .
\end{align}
Hence, for any $u,v$,
\begin{align}
p_{U,V}(u,v) & \le p_U(u)p_V(v)+\epsilon\\
& = p_U(u)p_V(v) \big (1+\frac{\epsilon}{p_U(u)p_V(v)} \big )\\
& \le p_U(u)p_V(v) \big (1+\frac{\epsilon}{\min\limits_u p_U(u)\times \min\limits_v p_V(v)} \big )\\
&= {\mu} \times p_U(u)p_V(v),
\end{align}
which completes the proof.
\end{IEEEproof}

\begin{lemma}\label{lemma7}
For any (discrete) random variables $U,V$ such that $V$ takes values from the set $\{0,1\}$,
\begin{align}
\| p_{U,V}-p_Up_V \|_1  \le 2 \times \max\limits_{u} \Big|p_{V|U}(1,u) - p_V(1)\Big|.
\end{align}
\end{lemma}

\begin{IEEEproof}
Note that we have
\begin{align}
\| p_{U,V}-&p_Up_V \|_1  = \sum_{u,v} \Big |p_{U,V}(u,v) - p_U(u)p_V(v)\Big | \\
&= \sum_{u,v} \Big |p_{U}(u)p_{V|U}(v,u) - p_U(u)p_V(v)\Big |\\
&\le  \sum_{u,v} p_{U}(u) \Big |p_{V|U}(v,u) - p_V(v) \Big |\\
&= \sum_{u} p_{U}(u) \Big |p_{V|U}(1,u) - p_V(1) \Big |\\
&+\sum_{u} p_{U}(u) \Big |p_{V|U}(0,u) - p_V(0) \Big |\\
&= \sum_{u} p_{U}(u) \Big |p_{V|U}(1,u) - p_V(1) \Big |\\
&+\sum_{u}  p_{U}(u) \Big |(1 -p_{V|U}(1,u)) - (1 -  p_V(1)) \Big |\\
&= 2 \times \sum_{u} p_{U}(u) \Big |p_{V|U}(1,u) - p_V(1) \Big |\\
&\le 2  \times \max\limits_{u} \Big |p_{V|U}(v,u) - p_V(v)\Big | \\
& \times \sum_{u} p_{U}(u)\\
&= 2 \times \max\limits_{u} \Big |p_{V|U}(1,u) - p_V(1) \Big |.
\end{align}
\end{IEEEproof}

 Concentration inequalities play an important role in the proofs of this paper.
Next we  state the Hoeffding's inequality.
\begin{lemma}\label{he} \cite[Theorem D.1]{fondml} 
(Hoeffding's inequality) 
Let $U_i$, $i \in [n]$, be $n$ i.i.d. random variables taking values in $[a,b]$.
Then, for any positive $\epsilon$, we have
\begin{align}
\pbb \Big ( \Big |\frac{U-\mathbb{E} [U]}{n} \Big | \ge \epsilon \Big ) \le 2  \exp \Big ( -2n\epsilon^2/(b-a)^2 \Big ),
\end{align}
where $U = \sum_{i=1}^n U_i$.
\end{lemma}

In what follows, we propose a lemma about the approximation of probabilities.  

\begin{lemma}\label{app}
Consider $n$ random variables $U_i$, $i \in [n]$,  which are distributed according a probability measure $p_{U_{1:n}}$, each over a finite set $\mathcal{U}$.
Also consider a discrete random variable $W$ which  takes values from a finite set $\mathcal{W}$. 
Let  $w^* := \argmax\limits_{w \in \mathcal{W}}  \pbb(W= w)$. 
Consider $n$ random variables $V_i$, $i \in [n]$, each takes values from  a finite set $\mathcal{V}=\mathcal{U}$,  such that
$
p_{V_{1:n}} = p_{U_{1:n}|W=w^*}\label{eq1}.
$
All in all, the probability low governing the above random variables factors as
\begin{align}
p_{U_{1:n},V_{1:n},W}=p_{U_{1:n}} p_{V_{1:n}} p_{W|U_{1:n}}.
\end{align}
In addition, consider an arbitrary deterministic function $\psi : \mathcal{U}^n  \to \{0,1\}$, and define the   events
$
\mathcal{E}_1 := \Big \{ \psi(U_{1:n}) = 1 \Big \},
$
 and
$
\mathcal{E}_2  := \Big \{ \psi(V_{1:n}) = 1\Big \}.
$
Then, 
\begin{align}
\pbb(\mathcal{E}_2) \le  |\mathcal{W}| \times \pbb(\mathcal{E}_1).
\end{align}
\end{lemma}

\begin{IEEEproof}
We write
\begin{align}
\pbb(\mathcal{E}_1) &= \mathbb{E} \Big [ \psi(U_{1:n}) \Big ] \\
& \overset{(a)}{=}   \mathbb{E}_W \mathbb{E} \Big [\psi(U_{1:n}) \Big |W \Big] \\
& \ge \pbb(W = w^*) \times   \mathbb{E} \Big [\psi(U_{1:n}) \Big |W = w^* \Big] \\
&= \pbb(W = w^*) \times   \mathbb{E}   \Big [ \psi(V_{1:n}) \Big ] \\
& \overset{(b)}{\ge} \frac{1}{|\mathcal{W}|} \times   \mathbb{E}   \Big [ \psi(V_{1:n})  \Big ]  \\
& = \frac{1}{|\mathcal{W}|} \times   \pbb(\mathcal{E}_2),
\end{align}
where (a) follows from the law of iterated expectation
and  (b) follows from the definition of $w^*$. 
\end{IEEEproof}


\begin{lemma}\label{lemma10}
For   given positive integers $n,m$ such that $n \ge m$,  define  
\begin{align}
\mathcal{Q}_m :=\Big \{ (v_1,v_2,\ldots,v_n) \in \{0,1\}^n :  \sum_{i=1}^n v_i =m \Big \}.
\end{align}
Consider $n$  binary random variables $V_i \in \{0,1\}$, $i \in [n]$,  which  are distributed as  
 \begin{align}
p_{V_{1:n}}(v_{1:n}) = \frac{1}{|\mathcal{Q}_m|} \times  \mathbbm{1} \{v_{1:n} \in \mathcal{Q}_m \}.
\end{align}
Then, for any positive $\epsilon$ and any (non-empty) $\mathcal{T} \subseteq [n]$, we have
\begin{align}
\pbb \Big ( \Big |\frac{V-\mathbb{E} [V]}{|\mathcal{T}|} \Big | \ge \epsilon \Big ) \le 2  (n+1)  \exp \Big ( -2|\mathcal{T}|\epsilon^2 \Big ),
\end{align}
where $V := \sum_{i\in \mathcal{T}} V_i$.
\end{lemma}

\begin{IEEEproof}
The proof is based on   Lemma \ref{he} and Lemma \ref{app}. 
Consider $n$ i.i.d. binary random variables $U_i \in \{0,1\}$, $i \in [n]$, such that $p:=\pbb(U_i = 1) =  m/n$.
Let us define a random variable $W  := \sum_{i \in [n]} U_i$.
Note that $W$ takes values from the set $\mathcal{W} = \{0,1,\ldots,n\}$ and it is distributed according to a binomial distribution  with parameters $n,p$. 
Note that 
\begin{align}
w^* &= \argmax\limits_{w \in \{0,1,\ldots, n \}}  \pbb(W= w) \\
& = \argmax_{w \in \{0,1,\ldots, n \}} { n \choose w}\times  p^{w} \times (1-p)^{n-w} \\
&\overset{(a)}{=} m, 
\end{align} 
where (a) follows since $m=np$ is   the mode of $W$. 

Now observe that for any $v_{1:n} \in \{0,1\}^n$ we have
\begin{align}
p_{U_{1:n}|W = m}&(v_{1:n}) \\
&=
\frac{
\pbb(W= m|U_{1:n}= v_{1:n}) \times \pbb( U_{1:n}=v_{1:n})
} {\pbb(W=m)}\\
& = \frac{
\mathbbm{1} \{v_{1:n} \in \mathcal{Q}_m \} \times p^{m} \times (1-p)^{n-m}
}{{n \choose m}\times  p^{m} \times (1-p)^{n-m}}\\
& \overset{(a)}{=} \frac{1}{{n \choose m}} \times \mathbbm{1} \{v_{1:n} \in \mathcal{Q}_m \} \\
& = p_{V_{1:n}}(v_{1:n}),
\end{align}
where (a) follows from the fact that $|\mathcal{Q}_m| = {n \choose m}$.
Let us define a  function $\psi:\{0,1\}^n \to \{0,1\}$ as follows.
\begin{align}
\psi (x_1,x_2,\ldots,x_n) = \mathbbm{1} \Big  \{
\Big |\frac{\sum_{i \in \mathcal{T}} x_i}{|\mathcal{T}|} - p  \Big | \ge \epsilon 
 \Big \}.
\end{align}
Now observe that the random variables $U_{1:n}$, $V_{1:n}$ and $W$ and the function $\psi(.)$ satisfy the required conditions of Lemma \ref{app}. Hence, if we define $U := \sum_{i\in \mathcal{T}} U_i$, we conclude that
\begin{align}
\pbb \Big ( \Big |\frac{V-\mathbb{E} [V]}{|\mathcal{T}|} \Big | \ge \epsilon \Big ) 
&\le |\mathcal{W}| \times \pbb \Big ( \Big |\frac{U-\mathbb{E} [U]}{|\mathcal{T}|} \Big | \ge \epsilon \Big ) \\
& \overset{(a)}{\le}  |\mathcal{W}| \times 2  \exp \Big ( -2|\mathcal{T}|\epsilon^2 \Big ) \\
& \overset{(b)}{=} 2  (n+1)  \exp \Big ( -2|\mathcal{T}|\epsilon^2 \Big ),
\end{align}
where (a) follows from Lemma \ref{he} and (b) follows the fact that $ \mathcal{W}  = \{0,1,\ldots,n\}$. We are done. 
\end{IEEEproof}

The following definitions are about the intersection of the subsequences.  
 \begin{definition}
 For any $\bold{s},\bold{t} \in \mathcal{S}_{L,G}$,  define 
  \begin{align}
\ints (\bold{s}, \bold{t}) := (w_1,w_2, \ldots, w_k )\in [G]^k
\end{align}• such that 
 \begin{itemize}
\item $\{w_1,w_2, \ldots, w_k\} = \{s_{\ell}:  \ell \in [L] \} \bigcap  \{t_{\ell}:  \ell \in [L] \}$,
\item $w_1 < w_2 < \ldots < w_k$.
\end{itemize}
  \end{definition}

  \begin{definition}
  For any $\bold{s} \in \mathcal{S}_{L,G}$, define
  \begin{align}
\mathcal{I}_{\epsilon} (\bold{s}) := \Big \{  \ints (\bold{s}, \bold{t}) ~ \Big | ~   \bold{t} \in \mathcal{S}_{L,G} ,\dist(\bold{s},\bold{t}) \ge L\epsilon \Big \}.
\end{align}
  \end{definition}
  Note that we have $|\mathcal{I}_{\epsilon} (\bold{s}) | \le 2^L$. 
  
\begin{lemma}\label{lemmatau}
For any $\bold{s},\bold{t} \in \mathcal{S}_{L,G}$,
the following statements are equivalent.
\begin{itemize}
\item $\dist(\bold{s}, \bold{t}) \ge L \epsilon$,
\item $\len(\ints (\bold{s}, \bold{t})) \le L(1-\epsilon/2)$.
\end{itemize}  
\end{lemma}

\begin{IEEEproof}
Let
$\mathcal{S} : = \{s_1,s_2,\ldots,s_L\}$ and $\mathcal{T} : = \{t_1,t_2,\ldots,t_L\}$. Note that
\begin{align}
\dist(\bold{s}, \bold{t}) \ge L \epsilon & \Leftrightarrow 
 |\mathcal{S} \cup \mathcal{T}| -  |\mathcal{S} \cap \mathcal{T}| \ge L\epsilon\\
 & \Leftrightarrow
 |\mathcal{S}| + |\mathcal{T}| - 2|\mathcal{S} \cap \mathcal{T}| \ge L \epsilon \\
  & \Leftrightarrow
 2L- 2|\mathcal{S} \cap \mathcal{T}| \ge L \epsilon \\
   & \Leftrightarrow
 |\mathcal{S} \cap \mathcal{T}| \le L(1- \epsilon/2) \\
    & \Leftrightarrow
 \len(\ints (\bold{s}, \bold{t})) \le L(1-\epsilon/2).
\end{align}
\end{IEEEproof}

\subsection{
 Preliminaries for Theorem \ref{lemma2}
}

In this subsection, we present a few lemmas and definitions related to the proof of Theorem  \ref{lemma2}.

 \begin{definition}
 For any  $g(.) \in \mathcal{F}_{L,m}$ and any  $\bold{t} \in \mathcal{S}_{L,G}$  define a random variable $J_{\bold{t},g} := g(\bold{X}_{\bold{t}})$.  Here $\bold{X}$ is a random sequence distributed uniformly over $\mathcal{X}^G$.
 \end{definition}

We note that the probability distribution of $J_{\bold{t},g}$ for any $g(.)  \in \mathcal{F}_{L,m}$ and $\bold{t}$ is as follows.
\begin{equation}
    p_{J_{\bold{t},g}}(u)= \begin{cases}
               \gamma               & u = 1\\
               1-\gamma            & u = 0,
           \end{cases}
\end{equation}
where $\gamma :=m/|\mathcal{X}|^L$ is a parameter. 
Also, for any function $F(.)$ which is chosen randomly and uniformly from the set $\mathcal{F}_{L,m}$ and any sequence $\bold{S}$ which is chosen randomly and uniformly from the set $\mathcal{S}_{L,G}$, we have the following identity.
\begin{align}
p_{J_{\bold{t},g}} = p_{J_{\bold{S},g}} = p_{J_{\bold{t},F}} = p_{J_{\bold{S},F}}.
\end{align}

We note that throughout this section, we have fixed an arbitrary $\bold{s}\in \mathcal{S}_{L,G}$ (see   Theorem \ref{lemma2} again for more information). 

In the following lemma, we aim to show that with probability tending one,  $J_{\bold{t},g}$ and $Y$ are (approximately) independent, if $\dist(\bold{s},\bold{t}) \ge L\epsilon$. 
In other words, we want to show that  $J_{\bold{t},g}~\overset{{\mu}}{\bot}~ Y$ for some ${\mu} \to 1$.
Let us clarify this statement in the following lemma.

 \begin{lemma}\label{lemma11}
 For any  $g(.) \in \mathcal{F}_{L,m}$,  any  $\bold{t} \in \mathcal{S}_{L,G}$ and also any $\mu \in (1, \infty)$ define the following event\footnote{
  Note that the  joint   distribution 
  $P_{Y,J_{\bold{t},g}}$ depends on the random function $F(.)$. Therefore,  it is a random pmf and we denote it  by capital letters. 
  As a reminder,  $F(.)$ is uniformly distributed over $\mathcal{F}_{L,m}$.
    }
  \begin{align}
\mathcal{E}^{\mu}_{\bold{t},g}:= \Big \{
J_{\bold{t},g}~\overset{{\mu}}{\bot}~ Y
 \Big \}.
\end{align}
Note that $Y = J_{\bold{s},F} \oplus Z$, where $Z$ is the additive noise in the model (see  Fig. 2). 
Let us define 
\begin{align}
\mathcal{E}^{\mu} := \bigcap_{
\substack{
            \bold{t} \in \mathcal{S}_{L,G}\\
             \dist(\bold{s},\bold{t}) \ge L \epsilon \\
             g(.) \in \mathcal{F}_{L,m}}
} \mathcal{E}^{\mu}_{ \bold{t} , g}
\end{align}
Then, $\pbb (\mathcal{E}^{\mu}) \to 1$ for any $\mu \in (1,\infty).$
 \end{lemma}

\begin{IEEEproof}
Let  us first define   
\begin{align}
\kappa := \frac{1}{2} (\mu - 1) \times \min (\beta,1-\beta) \times \min (\gamma, 1-\gamma ).
\end{align}
Assume that $\bold{X}$ is chosen randomly and uniformly from the set $\mathcal{X}^L$.
Let us define the  event\footnote{
Note that the random choice of the function $F(.)$ does  not make any difference in the  distribution of $J_{\bold{s},f}$, i.e.,  
$P_{J_{\bold{s},F}} = p_{J_{\bold{s},f}}$.
}
\begin{align}
\mathcal{E}^{\kappa}_{\bold{w},\bold{x}'} := \Big \{ \Big | P_{J_{\bold{s},F}|\bold{X}_{\bold{w}} = \bold{x}'} (1) -    P_{J_{\bold{s},F}}(1)\Big | \le  \kappa \Big   \},
\end{align}
for any $\bold{w} \in \mathcal{I}_{\epsilon}(\bold{s})$ and any $\bold{x}'  \in \mathcal{X}^{\len(\bold{w})}$.
Now we need the following two lemmas.
 \begin{lemma}\normalfont \label{appendlemma1}
\begin{align}
\mathcal{E}_{\text{int}}^{\kappa}:&=
\bigcap_{
\substack{
            \bold{t} \in \mathcal{S}_{L,G}\\
             \dist(\bold{s},\bold{t}) \ge L \epsilon \\
             \bold{x}'  \in \mathcal{X}^{\len(\ints (\bold{s}, \bold{t}))} }
}  \mathcal{E}^{\kappa}_{\ints (\bold{s}, \bold{t}),\bold{x}'} \\&
 \subseteq 
\mathcal{E}^{\mu} = \bigcap_{
\substack{
            \bold{t} \in \mathcal{S}_{L,G}\\
             \dist(\bold{s},\bold{t}) \ge L \epsilon \\
             g(.) \in \mathcal{F}_{L,m}}
} \mathcal{E}^{\mu}_{ \bold{t} , g} \label{intercup}
\end{align}
\end{lemma}
\begin{IEEEproof}
See appendix \ref{appendlemma}
\end{IEEEproof}
 \begin{lemma}\label{lemma13}
For any $\bold{t} \in \mathcal{S}_{L,G}$, such that $\dist(\bold{s},\bold{t}) \ge L\epsilon$,  any $g(.) \in \mathcal{F}_{L,m}$, and any $\bold{x}' \in \mathcal{X}^{\len(\ints (\bold{s}, \bold{t}))} $ we have
\begin{align}
\pbb( \overline{\mathcal{E}^{\kappa}}_{\ints (\bold{s}, \bold{t}),\bold{x}'})   \le
2   (|\mathcal{X}|^L+1)  \exp \Big ( -2
|\mathcal{X}|^{\frac{1}{2} L \epsilon}\kappa^2 \Big ).
\end{align}
\end{lemma}
\begin{IEEEproof}
See appendix \ref{appendlemma2}.
\end{IEEEproof}
Now using the  union  bound, we have
\begin{align}
\pbb (\mathcal{E}^{\mu}) & \overset{(a)}{\ge} \pbb(\mathcal{E}_{\text{int}}^{\kappa}) \\
&= 1 - 
\pbb(\overline{\mathcal{E}_{\text{int}}^{\kappa}}) \\
& = 1 -  \pbb \Big (\bigcup_{
\substack{
            \bold{t} \in \mathcal{S}_{L,G}\\
             \dist(\bold{s},\bold{t}) \ge L \epsilon \\
             \bold{x}' \in \mathcal{X}^{\len(\ints (\bold{s}, \bold{t}))} }
}  \overline{\mathcal{E}^{\kappa}}_{\ints (\bold{s}, \bold{t}),\bold{x}'} \Big) \\
 &= 1 -  \pbb \Big (\bigcup_{
\substack{
             \bold{w} \in \mathcal{I}_{\epsilon}(\bold{s})\\
             \bold{x}' \in \mathcal{X}^{\len(\bold{w})} }
}  \overline{\mathcal{E}^{\kappa}}_{\bold{w},\bold{x}'} \Big) \\
&\ge 1 -   \sum_{
\substack{
			\bold{w} \in \mathcal{I}_{\epsilon}(\bold{s})\\
             \bold{x}' \in \mathcal{X}^{\len(\bold{w})} }
}  \pbb (  \overline{\mathcal{E}^{\kappa}}_{\bold{w},\bold{x}'} ) \\
&\overset{(b)}{\ge} 1 -   \sum_{
\substack{
			\bold{w} \in \mathcal{I}_{\epsilon}(\bold{s})\\
             \bold{x}' \in \mathcal{X}^{\len(\bold{w})} }
} 2   (|\mathcal{X}|^L+1)   \exp \Big ( -2
|\mathcal{X}|^{\frac{1}{2} L \epsilon}\kappa^2 \Big  ) \\
& \overset{(c)}{\ge}  1 -   
  |\mathcal{X}|^{L}
  2^{L+1}   
   (|\mathcal{X}|^L+1)   \exp \Big ( -2
|\mathcal{X}|^{\frac{1}{2} L \epsilon}\kappa^2 \Big  ) \\
& = 1- o(1),
\end{align}
where (a) follows from Lemma \ref{appendlemma1},  (b) follows from Lemma \ref{lemma13}, and (c) follows from the fact that  $|\mathcal{I}_{\epsilon} (\bold{s}) | \le 2^L$.
The proof is thus complete.

\subsection{
Proof of Theorem \ref{lemma2}
}
Now we are ready to prove Theorem \ref{lemma2}.  
Fix an arbitrary $\bold{s} \in \mathcal{S}_{L,G}$ throughout the proof. 
Note that 
\begin{align}
 {I}(F(\bold{X}_{\bold{s}});Y)  & = {I}(J_{\bold{s},F};Y) = H(Y) - H(Y | F(\bold{X}_{\bold{s}}) ) \\
 &= h(\beta) - h(\alpha) .
\end{align}

Based on the joint AEP theorem \cite[Theorem 7.6.1]{cover}, if   $f(\bold{X}_{\bold{s}})$ is independent of  $g(\bold{X}_{\bold{t}})$, for any $f(.),g(.) \in \mathcal{F}_{L,m}$ and  any $\bold{t} \in \mathcal{S}_{L,G}$ with the normalized distance of at least $\epsilon$ from $\bold{s}$, then the desired result is established.
 However, in the theorem, this condition does not hold.
 In particular, if two sequences $\bold{s},\bold{t}$ intersect, then the independence  may not hold.
 This  means that we cannot immediately use the AEP theorem   for the proof. 
However, we showed that if we choose the function $F(.) \in \mathcal{F}_{L,m}$ uniformly at random, then for  sequences like $\bold{t}$ that have at least a normalized distance of $\epsilon$ from $\bold{s}$, the independence condition holds approximately. 
We proved this statement in Lemma \ref{lemma11}.

Now we  use similar steps to \cite[Theorem 7.6.1]{cover}  to prove the theorem. 
Fix an arbitrary $\mu \in (1,\infty)$.
Consider the event $\mathcal{E}^{\mu}$
which is defined in Lemma \ref{lemma11}.
  Fix  a sequence $\bold{t} \in \mathcal{S}_{L,G}$ such that $\dist(\bold{s},\bold{t}) > L\epsilon$, and  a function $g(.) \in \mathcal{F}_{L,m}$. Let  $U_n := g(\bold{X}_{n,\bold{t}})
  $
  for any $n \in [N]$. 
  
   Note that  
   \begin{align}
\pbb(&\mathcal{E}_{\bold{t},g}|\mathcal{E}^{\mu}) \\& = \pbb \Big \{ \Big ( 
\big (U_n)_{n \in [N]},(Y_n)_{n \in [N]}\Big ) \in \mathcal{T}^N_{\tau}
\Big |
 \mathcal{E}^{\mu},\bold{S} = \bold{s}
 \Big \}  \\
& = \sum_{(u^N,y^N) \in \mathcal{T}_{\tau}^N}  p_{U^N,Y^N|\mathcal{E}^{\mu},\bold{S} = \bold{s}} (u^N,y^N).
\end{align}
Using Lemma \ref{lemma11}, we conclude that
\begin{align}
&\pbb(\mathcal{E}_{\bold{t},g} | \mathcal{E}^{\mu})  =  \sum_{(u^N,y^N) \in \mathcal{T}_{\tau}^N}  p_{U^N,Y^N|\mathcal{E}^{\mu},\bold{S} = \bold{s}} (u^N,y^N)\\
&  
  \le  \mu^N  \times \sum_{(u^N,y^N) \in \mathcal{T}_{\tau}^N} p_{U^N|\mathcal{E}^{\mu},\bold{S} = \bold{s}}(u^N) \times  p_{{Y}^N|\mathcal{E}^{\mu},\bold{S} = \bold{s}}(y^N) 
\\
& \le  \mu^N \times  \sum_{(u^N,y^N) \in \mathcal{T}_{\tau}^N} 2^{-NH(g(\bold{X}_\bold{t})) - NH(Y)+2N\tau }\\
& \le  \mu^N \times  2^{NH(g(\bold{X}_\bold{t}),Y)-NH(g(\bold{X}_\bold{t})) - NH(Y)+3N\tau }\\
&=  2^{-N(h(\beta) - h(\alpha)-\log(\mu)-3\tau) }
\end{align}
which completes the proof, if $\mu \in (1,\infty)$ is  small enough.

\section{Proof of Lemma \ref{lemm1}}\label{prooflemm3}

Let us first state some preliminaries. Define $\delta(\epsilon) :
= \sup\limits_{x \in (0,1/2)} \frac{h(2\epsilon x)}{h(x)}.$
It can be shown that  $\delta(\epsilon)=h(\epsilon)$ for any $\epsilon \in (0,1/2).$ This follows from the following lemma.

\begin{lemma}\label{lemma4}
For any $x,y \in [0,1/2]$ we have
\begin{align}
h(2xy) \le h(x) h(y),
\end{align}
where $h(.)$ is the binary entropy function.
\end{lemma}

\begin{IEEEproof}
See appendix \ref{appendD}.
\end{IEEEproof}

Using   Lemma \ref{lemma4}, we conclude that 
\begin{align}
h(\epsilon) = \frac{h(2\epsilon x)}{h(x)} \Big |_{x = 1/2} &\le \sup_{x \in (0,1/2)} \frac{h(2\epsilon x)}{h(x)}\\
& \le \sup_{x \in (0,1/2)} h(\epsilon) = h(\epsilon),
\end{align}
which shows that $\delta(\epsilon) = h(\epsilon)$. Next we  use the function $\delta(.)$ to achieve the desired result.

By the  assumption of the lemma, there is  a sequence of algorithms  $\{\mathcal{A}_{(G_i,L_i,N_i,{\alpha}, \beta ,m_i)}\}_{i\in \mathbb{N}}$  with rate $R$, such that we have $\lim_{i \rightarrow \infty}P^{\text{WC}}_{\epsilon}(\mathcal{A}_{G_i})=0.$ 
This implies that $\lim_{i \rightarrow \infty}P^{\text{AVG}}_{\epsilon}(\mathcal{A}_{G_i})=0.$ 
 For a fixed positive integer $i$, let $\bold{S}$ be a random sequence that  is distributed   uniformly over the set $\mathcal{S}_{L_i,G_i}$. 
 Also let $F(.)$ be a random function that is distributed uniformly   over the set $\mathcal{F}_{L_i,m_i}$. 
 There are $N_i$   samples  $\bold{X}^{N_i}$, which are sampled uniformly and independently from the set $|\mathcal{X}|^G$, and their corresponding labels $Y^{N_i}$, which are generated based on  the parameters $\alpha$ and $\beta$. We denote the output of the algorithm $\mathcal{A}_{(G_i,L_i,N_i,{\alpha}, \beta ,m_i)}$ to the dataset $(\bold{X}^{N_i},\bold{Y}^{N_i})$ by $\hat{\bold{S}}$.
Let us define the event $\mathcal{E}_i:=\mathcal{E}_{\mathcal{A}_{G_i},\epsilon} = \{\frac{\dist(\hat{\bold{S}},\bold{S})}{L_i}>\epsilon\}$ and also let $E_i:=\mathbbm{1}\{\mathcal{E}_i\}$. Note that $\pbb(\mathcal{E}_{i}) = P^{\text{AVG}}_{\epsilon}(\mathcal{A}_{G_i}).$ Hence, we have $\lim_{i  \rightarrow \infty}\pbb(\mathcal{E}_{i})=0$.

The proof consists of the following  steps. 
\begin{enumerate}[(i)]
\item   First we claim that \begin{align}
H(\bold{S}|\hat{\bold{S}})& \le 1+\pbb(\mathcal{E}_i)\log({ G_i \choose L_i})+\log ( L_i \epsilon+1) \\
&+ L_i h(\epsilon)+ (G_i-L_i)h(\frac{L_i \epsilon }{G_i-L_i}).
\end{align}

\item  Therefore,  
\begin{align}
\log(&{G_i \choose L_i})\\&=H(\bold{S})  \\
&=H(\bold{S}|\hat{\bold{S}})+{I}(\bold{S};\hat{\bold{S}})  \\ 
&\le  1+\pbb(\mathcal{E}_i)\log({ G_i \choose L_i})+\log ( L_i \epsilon+1)  \\
& + L_i h(\epsilon) + (G_i-L_i)h(\frac{L_i \epsilon }{G_i-L_i}) + {I}(\bold{S};\hat{\bold{S}})\label{10}.
\end{align}

\item 
The third step   is to show that 
\begin{align}
{I}(\bold{S};\hat{\bold{S}})\le N_i h(\beta)-N_i h(\alpha)\label{91}.
\end{align}

\item Combining the above arguments shows that
\begin{align}
\log(&{G_i \choose L_i})\\&\le 1+\pbb(\mathcal{E}_i)\log({ G_i \choose L_i})+\log (L_i \epsilon+1)\\
&+L_i h(\epsilon) +(G_i-L_i)h(\frac{\epsilon L_i}{G_i-L_i}) \\
&+N_i h(\beta)-N_i h(\alpha)\label{11}.
\end{align}

\item

Then, from \cite[Chapter  11, p.~353]{cover},  
\begin{align}
 \frac{1}{G_i+1}2^{G_i h(L_i/G_i)} \le {G_i \choose L_i} \le 2^{G_i h(L_i/G_i)}.
\end{align}
By taking the logarithm from the two sides,  we conclude
\begin{align}
 G_i h(L_i/G_i)-\log(G_i+1)&\le  \log({G_i \choose L_i})\\
 &\le G_i h(G_i/L_i)\label{12}.
\end{align}

\item 

Using (\ref{12}) and (\ref{11}), we have
\begin{align}
&G_i h(L_i/G_i) - \log(G_i+1) \\
&\le 1+\pbb(\mathcal{E}_i) G_i h(L_i/G_i)+\log (L_i \epsilon+1) +L_i h(\epsilon) \\
&+(G_i-L_i)h(\frac{\epsilon L_i}{G_i-L_i})+N_i h(\beta)-N_i h(\alpha)\label{13}.
\end{align}
Dividing two sides of  (\ref{13}) by $G_i h(L_i/G_i)$ results
\begin{align}
 &1 -\frac{ \log(G_i+1)}{G_i h(L_i/G_i)}\\ &\le \frac{1}{G_i h(L_i/G_i)}+\pbb(\mathcal{E}_i) +\frac{\log (L_i\epsilon+1)}{G_i h(L_i/G_i)} \\
 &+\frac{L_i }{G_i h(L_i/G_i)} h(\epsilon)+\frac{(G_i-L_i)}{G_i h(L_i/G_i)}h(\frac{\epsilon L_i}{G_i-L_i})\\&+\frac{h(\beta)-h(\alpha)}{R}.
\end{align}
Using the concavity of the function $h(.)$, we have that 
\begin{align}
\frac{L_i }{G_i} h(\epsilon)+\frac{(G_i-L_i)}{G_i }h(\frac{\epsilon L_i}{G_i-L_i}) \le h (2\epsilon L_i/G_i).
\end{align}
Hence, we conclude that
\begin{align}
 1 -\frac{ \log(G_i+1)}{G_i h(L_i/G_i)}&\le \frac{1}{G_i h(L_i/G_i)}+ \pbb(\mathcal{E}_i)\\& +\frac{\log (L_i\epsilon+1)}{G_i h(L_i/G_i)} +\frac{h(2\epsilon L_i/G_i)}{h(L_i/G_i)} \\
 &+\frac{h(\beta)-h(\alpha)}{R} .
\end{align}
Applying the inequality  $G_i h(L_i/G_i) \ge L_i \log(G_i/L_i)$ shows that  
\begin{align}
 1 -\frac{ \log(G_i+1)}{L_i \log(G_i/L_i)}&\le \frac{1}{G_i h(L_i/G_i)}+\pbb(\mathcal{E}_i) \\
 &+\frac{\log (L_i \epsilon+1)}{L_i \log(G_i/L_i)}+\frac{h(2\epsilon L_i/G_i)}{h(L_i/G_i)} \\
 &+\frac{h(\beta)-h(\alpha)}{R} \\
 &\overset{(a)}{\le} \frac{1}{G_i h(L_i/G_i)}+\pbb(\mathcal{E}_i) \\
 & +\frac{\log (L_i \epsilon+1)}{L_i \log(G_i/L_i)}  \\
 &+\delta (\epsilon)+\frac{h(\beta)-h(\alpha)}{R} \label{14},
\end{align}  where $(a)$ follows by the definition of the function $\delta(.)$.

\item  Finally, we claim that at the limit of $i \rightarrow \infty$, two terms $\frac{ \log(G_i+1)}{L_i \log(G_i/L_i)}$ and $\frac{\log (L_i \epsilon+1)}{L_i \log(G_i/L_i)}$  go to zero. By letting  $i \rightarrow \infty$ in (\ref{14}),  we conclude that
\begin{align}
1 \le \delta (\epsilon)  + \frac{h(\beta)-h(\alpha)}{R},
\end{align}
or 
\begin{align}
R \le \frac{h(\beta)-h(\alpha)}{1-\delta(\epsilon)},
\end{align}
which completes the  proof of the lemma.

\end{enumerate}

In what follows, we prove the above stated claims. In particular, the claims in steps (i),(ii) and (vii) need to be proved.

\subsection{Proof of (i)}

Note that we have
\begin{align}
H(E_i,\bold{S}|\hat{\bold{S}})&=H(\bold{S}|\hat{\bold{S}})+H(E_i|\bold{S},\hat{\bold{S}})   \label{7}\\
&=H(E_i|\hat{\bold{S}})+H(\bold{S}|\hat{\bold{S}},E_i) .
\end{align}
Note that $H(E_i|\bold{S},\hat{\bold{S}})=0$ and  $H(E_i|\hat{\bold{S}})\le 1$. Therefore, 
\begin{align}
&H(\bold{S}|\hat{\bold{S}}) \\&\le 1+ H(\bold{S}|\hat{\bold{S}},E_i)\\
&= 1+\pbb(\mathcal{E}_i)H(\bold{S}|\hat{\bold{S}},E_i=1) + (1-\pbb(\mathcal{E}_i))H(\bold{S}|\hat{\bold{S}},E_i=0)\\
&\le   1+\pbb(\mathcal{E}_i)H(\bold{S}|\hat{\bold{S}},E_i=1) +H(\bold{S}|\hat{\bold{S}},E_i=0)\\
	&\overset{(a)}{\le}   1+\pbb(\mathcal{E}_i)H(\bold{S}) +H(\bold{S}|\hat{\bold{S}},E_i=0)\\
&=  1+\pbb(\mathcal{E}_i)\log({ G_i \choose L_i}) +H(\bold{S}|\hat{\bold{S}},E_i=0)\label{9},
\end{align}
where (a) follows from the fact that conditioning reduces the entropy. Note that
\begin{align}
&H(\bold{S}|\hat{\bold{S}},E_i=0)\\&\le \max_{\hat{\bold{s}}\in \mathcal{S}_{L_i,G_i}}\log(|\{ \bold{t} \in \mathcal{S}_{L_i,G_i}: \frac{\dist( \hat{\bold{s}},\bold{t})}{L_i}\le \epsilon\}|)\\
&= \max_{\hat{\bold{s}}\in \mathcal{S}_{L_i,G_i}}  \log(\sum_{\ell=0}^{\lfloor L_i \epsilon \rfloor}|\{  \bold{t}  \in \mathcal{S}_{L_i,G_i}: \dist(\hat{\bold{s}}, \bold{t})=\ell \}|) \\
&=\max_{\hat{\bold{s}}\in \mathcal{S}_{L_i,G_i}}  \log(\sum_{\ell=0}^{\lfloor L_i \epsilon \rfloor}{ L_i \choose \ell}{G_i-L_i \choose \ell})\\
&\overset{(a)}{\le}  \log( (L_i \epsilon+1){ L_i \choose \lfloor L_i\epsilon \rfloor}{G_i-L_i\choose \lfloor  L_i  \epsilon \rfloor}),\label{81}
\end{align}
where (a) follows from the fact that $\epsilon \in (0, 1/2)$. 
Using (\ref{12}) and (\ref{81}) we conclude
\begin{align}
&H(\bold{S}|\hat{\bold{S}},E_i=0) \\&\le  \log( (L_i \epsilon + 1){ L_i  \choose \lfloor L_i\epsilon \rfloor}{G_i-L_i \choose \lfloor  L_i \epsilon \rfloor})\\
&  \le  \log( (L_i \epsilon +1) \times 2 ^{ L_i h(\epsilon)}\times  2^{(G_i-L_i)h(\frac{L_i \epsilon }{G_i-L_i})})\\
&= \log ( L_i \epsilon+1)+ L_i h(\epsilon)+ (G_i-L_i)h(\frac{L_i \epsilon }{G_i-L_i})\label{8},
\end{align}
Combining (\ref{9}) and (\ref{8}) results
\begin{align}
H(\bold{S}|\hat{\bold{S}})&\le 1+\pbb(\mathcal{E}_i)\log({ G_i \choose L_i})+\log ( L_i \epsilon+1)+ L_i h(\epsilon)\\&+ (G_i-L_i)h(\frac{L_i \epsilon }{G_i-L_i}),
\end{align}
which completes the proof.

\subsection{Proof of (iii)}
Note that \begin{align}
{I}(\bold{S};\hat{\bold{S}})& \overset{(a)}{\le} {I}(\bold{S};\bold{X}^{N_i},Y^{N_i})\\
&=  {I}(\bold{S};\bold{X}^{N_i})+{I}(\bold{S};Y^{N_i}|\bold{X}^{N_i})\\
&\overset{(b)}{=} {I}(\bold{S};Y^{N_i}|\bold{X}^{N_i})\\
&=H(Y^{N_i}|\bold{X}^{N_i})-H(Y^{N_i}|\bold{X}^{N_i},\bold{S})\\
&\overset{(c)}{\le} N_i  h(\beta)-H(Y^{N_i}|\bold{X}^{N_i},\bold{S})\label{2},
\end{align}
where  (a) follows from  the data processing inequality and the fact that $\bold{S} \rightarrow (\bold{X}^{N_i},Y^{N_i}) \rightarrow \hat{\bold{S}}$ is a Markov Chain,  (b) follows from the fact that $\bold{S}$ and $\bold{X}^{N_i}$ are independent random variables and thus ${I}(\bold{S};\bold{X}^{N_i})=0$, and (c) follows from the fact that  
\begin{align}
 H(Y^{N_i}|\bold{X}^{N_i})\le H(Y^{N_i})= N_i H(Y) &= N_i h(\pbb(Y=1))\\&=N_i h(\beta).
\end{align} 
We write
\begin{align}
H(Y^{N_i}|\bold{X}^{N_i},\bold{S})&\overset{(a)}{\ge} H(Y^{N_i}|\bold{X}^{N_i},\bold{S},F(.))  \\
&= H(Y^{N_i}|(F(\bold{X}_{n ,{\bold{S}}}))_{n \in [N_i]})   \\
& \overset{(b)}{=}\sum_{n=1}^{N_i} H(Y_n|(F(\bold{X}_{n ,{\bold{S}}}))_{n \in [N_i]},Y^{n-1})  \\
& \overset{(c)}{=}N_i H(Y_1|F(\bold{X}_{1,\bold{S}}))   \\
& = N_i h(\alpha) \label{1},
\end{align}
where (a) follows by the fact that conditioning reduces the entropy, (b) follows by the telescopic property of the joint entropy and (c) holds because of the memoryless property of  the additive noise in the model.  Combining (\ref{2}) and (\ref{1}) results 
\begin{align}
{I}(\bold{S};\hat{\bold{S}})\le N_i h(\beta)-N_i h(\alpha)\label{91},
\end{align}
which completes the proof.

\subsection{Proof of (vii)}

 To prove the claim, it suffices to only show that $\frac{ \log(G_i+1)}{L_i \log(G_i/L_i)}$  vanishes asymptotically, since   $\log(G_i+1) \ge \log (L_i \epsilon+1)$.  We write
\begin{align}
\frac{\log(G_i+1)}{L_i \log(G_i/L_i)}& = \frac{\log(G_i+1)}{L_i \log(G_i) - L_i \log(L_i)}\\
&= \frac{1}{L_i (\frac{\log(G_i) }{\log(G_i+1)}-  \frac{\log(L_i)}{\log(G_i+1)})} \label{105}.
\end{align}
 Note that
\begin{align}
&L_i (\frac{\log(G_i) }{\log(G_i+1)}-   \frac{\log(L_i)}{\log(G_i+1)}) \\ &\ge  L_i (\frac{\log(G_i) }{\log(G_i+1)}-  \frac{\log(L_i)}{\log(2L_i+1)})   \\
&  =  L_i (\frac{\log(G_i) }{\log(G_i+1)} -  \frac{\log(2L_i+1)  + \log(\frac{L_i}{L_i + 1/2}) - 1}{\log(2L_i+1)})     \\
&= - L_i \frac{\log((G_i+1)/G_i)}{\log(G_i+1)}  + \Theta(\frac{L_i}{\log(L_i)})   \label{104}.
\end{align}
Note that
\footnote{
Note that for any constant $c \in (0,1)$, we have $c/G_i  \le \ln(1+1/G_i)  \le  1 /G_i$ for any  large enough $G_i$.
}  $L_i \frac{\log((G_i+1)/G_i)}{\log(G_i+1)} = O ( \frac{L_i}{G_i  \log(G_i+1)})$ and it vanishes asymptotically. Hence,  at the limit of $i \rightarrow \infty$, the R.H.S of (\ref{104}) is $\Theta(L_i/\log(L_i))$ and so we have $L_i (\frac{\log(G_i) }{\log(G_i+1)}- \frac{\log(L_i)}{\log(G_i+1)}) \rightarrow  \infty$ as $i \rightarrow \infty$. Therefore, using (\ref{105}), we conclude that $\frac{\log(G_i+1)}{L_i \log(G_i/L_i)}  \rightarrow 0$  as $i \rightarrow \infty$ and so the claim is proved.

\section{
Proof of Lemma \ref{lemma4}
}\label{appendD}

Let us define $f_y(x):= h(x)h(y) - h(2xy)$ for any $y \in [0,1/2]$.
Note that $f_y(0) = f_y(1/2)=0$. We aim to prove that $f_y(.)$ is a non-negative function on the interval $[0,1/2]$. To complete the proof, it suffices to show that for each $y$, the function $f_y(.)$ is concave on the interval $[0,1/2]$.

Note that 
\begin{align}
\frac{d}{dx} f_y(x) = \frac{h(y)}{\ln(2)} \ln(\frac{1-x}{x}) - \frac{2y}{\ln(2)} \ln(\frac{1-2xy}{2xy}). 
\end{align}
Now we write
\begin{align}
\frac{d^2}{dx^2} f_y(x)& = \frac{h(y)}{\ln(2)} \frac{-1}{x(1-x)} + \frac{4y^2}{\ln(2)}  \frac{1}{2xy(1-2xy)}  \\
& = \frac{1}{x\ln(2)} \Big (     \frac{-h(y)}{1-x} + \frac{2y}{1-2xy}      \Big)  \\
& \le \frac{1}{x(1-x)\ln(2)} \Big (     -h(y) + 2y \Big)  \\
& \le 0, 
\end{align}
which completes the proof, since   $h(y) \ge 2y$ for any $y \in [0,1/2]$.

\section{
Proof of Lemma \ref{appendlemma1}
}\label{appendlemma}

  Assume that the event in the L.H.S. of 
  (\ref{intercup}) happens. We aim to prove that for any $\bold{t} \in \mathcal{S}_{L,G}$, such that $\dist(\bold{s},\bold{t}) \ge L\epsilon$, and any $g(.) \in \mathcal{F}_{L,m}$, the event $\mathcal{E}^{\mu}_{ \bold{t} , g}$ occurs. 
We write
  \begin{align}
 &\mathcal{E}^{\mu}_{\bold{t},g} = \Big \{
J_{\bold{t},g}~\overset{{\mu}}{\bot}~ Y
 \Big \} \\
 & \overset{(a)}{\supseteq}   \Big \{    \|   P_{Y,J_{\bold{t},g}} - P_Y P_{J_{\bold{t},g}} \|_1   
 \\& \quad \quad \le  ({\mu - 1})\times {\min\limits_{y} P_Y(y) \times \min\limits_{u} P_{J_{\bold{t},g}}(u)}
  \Big \} \\
 & = \Big \{    \|   P_{Y,J_{\bold{t},g}} - P_Y P_{J_{\bold{t},g}} \|_1   
  \le 2 \times \kappa
  \Big \} \\
   & \overset{(b)}{\supseteq}   \Big \{    \|   P_{\bold{X}_{\ints (\bold{s}, \bold{t})},J_{\bold{s},F}} - P_{\bold{X}_{\ints (\bold{s}, \bold{t})}} P_{J_{\bold{s},F}} \|_1   
  \le 2 \times  \kappa
  \Big \} \\
     & \overset{(c)}{\supseteq}   \Big \{   
     \max_{\bold{x}' \in   \mathcal{X}^{\len(\ints (\bold{s}, \bold{t}))}} \Big |   P_{J_{\bold{s},F}| \bold{X}_{\ints (\bold{s}, \bold{t})} 
     = \bold{x}'
     }(1) -  P_{J_{\bold{s},F}(1)}
     \Big |
  \le \kappa
  \Big \} \\
  &= \bigcap_{\bold{x}' \in   \mathcal{X}^{\len(\ints (\bold{s}, \bold{t}))}}  \Big \{   
       \Big |   P_{J_{\bold{s},F}| \bold{X}_{\ints (\bold{s}, \bold{t})} 
     = \bold{x}'
     }(1) -  P_{J_{\bold{s},F}(1)}
     \Big |
  \le \kappa
  \Big \} \\
&= \bigcap_{\bold{x}' \in   \mathcal{X}^{\len(\ints (\bold{s}, \bold{t}))}}  
	\mathcal{E}^{\kappa}_{\ints (\bold{s}, \bold{t}),\bold{x}'}, \\ 
  \end{align}
where (a) follows from Lemma \ref{lemma6}, (b) follows from Corollary \ref{corollary1} and the fact that we have the following Markov chain
\begin{align}
J_{\bold{t},g} - \bold{X}_{\bold{t}} - \bold{X}_{\ints (\bold{s}, \bold{t})} - \bold{X}_{\bold{s}} - J_{\bold{s},F} - Y,
\end{align}
and also, (c) follows from Lemma \ref{lemma7}.
The proof is thus complete. 
\end{IEEEproof}

\section{
Proof of Lemma \ref{lemma13}
}\label{appendlemma2}
 
Let $\bold{X} \in \mathcal{X}^G$ be a random sequence which is distributed uniformly over $\mathcal{X}^G$.
Note that
\begin{align}
 P_{J_{\bold{s},F}|\bold{X}_{\ints (\bold{s}, \bold{t})} 
 = \bold{x}'} &(1)\\& =
 \mathbb{E}_{\bold{X} 
 } \Big [F(X_{\bold{s}}) \Big | \bold{X}_{\ints (\bold{s}, \bold{t})} = \bold{x}' , F\Big ]
  \\
 & \overset{(a)}{=} \frac{1}{|\mathcal{X}|^{(L - \len(\ints (\bold{s}, \bold{t})))}} \sum_{
 \substack{
              \bold{x} \in \mathcal{X}^G\\
               \bold{x}_{\ints (\bold{s}, \bold{t})} 
 = \bold{x}'}
 }  F(\bold{x}_{\bold{s}}),
\end{align}
where (a) follows from the fact that 
\begin{align}
\Big  | \Big\{  \bold{x} \in \mathcal{X}^G :  \bold{x}_{\ints (\bold{s}, \bold{t})} = \bold{x}'  \Big \} \Big | = |\mathcal{X}|^{(L - \len(\ints (\bold{s}, \bold{t})))}.
\end{align}
Now we write
\begin{align}
\pbb( &\overline{\mathcal{E}^{\kappa}}_{\ints (\bold{s}, \bold{t}),\bold{x}'})  \\&=  \pbb \Big ( \Big | P_{J_{\bold{s},F}|\bold{X}_{\ints (\bold{s}, \bold{t})} = \bold{x}'} (1) -    P_{J_{\bold{s},F}}(1)\Big | \ge \kappa  \Big ) \\
  & =  \pbb \Big ( \Big |
 P_{J_{\bold{s},F}|\bold{X}_{\ints (\bold{s}, \bold{t})} 
 = \bold{x}'} (1)
  -     \gamma \Big | \ge  \kappa  \Big ) \\
  &  =  \pbb \Big ( \Big |
  \frac{1}{|\mathcal{X}|^{(L - \len(\ints (\bold{s}, \bold{t})))}} \sum_{
 \substack{
              \bold{x} \in \mathcal{X}^G\\
               \bold{x}_{\ints (\bold{s}, \bold{t})} 
= \bold{x}'}
 }  F(\bold{x}_{\bold{s}})
  -     \gamma \Big | \ge  \kappa  \Big )\label{eq3},
\end{align}
%
%
Now let $\phi :   \mathcal{X}^L  \to [|\mathcal{X}|^L]$ be a bijection mapping and  $n = |\mathcal{X}|^L$. 
 Define $n$ random variables $V_i= F({\phi^{-1}(i)})$, $i \in [n]$.  
Let 
 \begin{align}
\mathcal{T} := \Big  \{ \phi({\bold{x}}) ~\Big |~ {\bold{x}} \in \mathcal{X}^G , {\bold{x}}_{\ints (\bold{s}, \bold{t})}
= \bold{x}'  \Big  \}.
\end{align}
Note that $|\mathcal{T}| = |\mathcal{X}|^{(L - \len(\ints (\bold{s}, \bold{t})))}$. Define  $V = \sum_{i \in \mathcal{T}} V_i$.
Note that $\mathbb{E} [V_i] =     \gamma$ for any $i$.

Note  that the random variables $V_i$, $i \in [n]$, and the set $\mathcal{T} \subseteq[n]$, satisfy the required conditions of Lemma \ref{lemma10}. 
Therefore,  using Lemma \ref{lemma10}, we conclude that 
\begin{align}
\pbb(& \overline{\mathcal{E}^{\kappa}}_{\ints (\bold{s}, \bold{t}),\bold{x}'})  
\\&  =
   \pbb \Big ( \Big |
  \frac{1}{|\mathcal{X}|^{(L - \len(\ints (\bold{s}, \bold{t})))}} \sum_{
 \substack{
              \bold{x} \in \mathcal{X}^G\\
               \bold{x}_{\ints (\bold{s}, \bold{t})} 
= \bold{x}'}
 }  F(\bold{x}_{\bold{s}})
  -     \gamma \Big | \ge  \kappa  \Big ) \\
& =\pbb \Big ( \Big |\frac{V-\mathbb{E} [V]}{|\mathcal{T}|} \Big | \ge \kappa \Big ) \\
&\le 2  (n+1)  \exp \Big ( -2|\mathcal{T}|\kappa^2 \Big )\\
&= 2  (|\mathcal{X}|^L+1)  \exp \Big ( -2
|\mathcal{X}|^{(L - \len(\ints (\bold{s}, \bold{t})))} \kappa^2 \Big ) \\
& \overset{(a)}{\le}2   (|\mathcal{X}|^L+1)  \exp \Big ( -2
|\mathcal{X}|^{\frac{1}{2} L \epsilon} \kappa^2 \Big )\label{eq2},
\end{align}
where (a) follows from Lemma \ref{lemmatau}.
The proof is thus complete.


\begin{IEEEbiographynophoto}{Behrooz Tahmasebi}
 is a Ph.D. student in Electrical Engineering and Computer Science (EECS) at MIT.
He  received his B.Sc. and M.Sc. degrees from the Department of Electrical Engineering, Sharif University of Technology, Tehran, Iran, in 2016 and 2018, respectively. His research interests include information theory, probability theory, and machine learning. 
\end{IEEEbiographynophoto}

\begin{IEEEbiographynophoto}{Mohammad Ali Maddah-Ali}  (Member, IEEE) received the B.Sc. degree in electrical engineering from the Isfahan University of Technology, the M.A.Sc. degree in electrical engineering from the University of Tehran, and the PhD degree from the Department of Electrical and Computer Engineering, University of Waterloo, Canada in 2007. From 2007 to 2008, he was with the Wireless Technology Laboratories, Nortel Networks, Ottawa, ON, Canada. From 2008 to 2010, he was a Post-Doctoral Fellow with the Department of Electrical Engineering and Computer Sciences, University of California at Berkeley. Then, he joined Nokia Bell Labs, Holmdel, NJ, USA, as a Communication Research Scientist. Recently, he started working at the Sharif University of Technology as a Faculty Member. He is a recipient of NSERC Postdoctoral Fellowship in 2007, the Best Paper Award from the IEEE International Conference on Communications (ICC) in 2014, the IEEE Communications Society and IEEE Information Theory Society Joint Paper Award in 2015, and the IEEE Information Theory Society Joint Paper Award in 2016.
\end{IEEEbiographynophoto}

\begin{IEEEbiographynophoto}{Seyed Abolfazl Motahari}
is an assistant professor at Computer Engineering Department of Sharif University of Technology. He received his B.Sc. degree from the Iran University of Science and Technology (IUST), Tehran, in 1999, the M.Sc. degree from Sharif University of Technology, Tehran, in 2001, and the Ph.D. degree from University of Waterloo, Waterloo, Canada, in 2009, all in electrical engineering. From August 2000 to August 2001, he was a Research Scientist with the Advanced Communication Science Research Laboratory, Iran Telecommunication Research Center (ITRC), Tehran. From October 2009 to September 2010, he was a Postdoctoral Fellow with the University of Waterloo, Waterloo. From 2010 to 2013, he was a Postdoctoral Fellow with the Department of Electrical Engineering and Computer Sciences, University of California at Berkeley. His research interests include multiuser information theory and Bioinformatics. He received several awards including Natural Science and Engineering Research Council of Canada (NSERC) Post-Doctoral Fellowship.
\end{IEEEbiographynophoto}

\end{document}